\documentclass[10pt, journal]{IEEEtran}
\IEEEoverridecommandlockouts
\usepackage[english]{babel}
\usepackage{blindtext}
\usepackage{algorithm}
\usepackage{algorithmic}
\usepackage{graphicx}
\usepackage{subfigure}
\usepackage{amsmath}
\usepackage{amssymb}
\usepackage{cases}
\usepackage{setspace}
\usepackage{multirow}
\usepackage{yhmath}
\newtheorem{theorem}{\textbf{Theorem}}
\newtheorem{lemma}{\textbf{Lemma}}
\newtheorem{definition}{\textbf{Definition}}

\begin{document}
\title{Learning-Based Coexistence in Two-Tier Heterogeneous Networks with Cognitive Small Cells}
\author{ \authorblockN{Lin Zhang, \IEEEmembership{Student Member, IEEE}, Guodong Zhao, \IEEEmembership{Member, IEEE}, Wenli Zhou, Gang Wu, \IEEEmembership{Member, IEEE}, Ying-Chang Liang, \IEEEmembership{Fellow, IEEE}, and Shaoqian Li, \IEEEmembership{Fellow, IEEE}}
\thanks{Lin Zhang, Guodong Zhao, Wenli Zhou, Gang Wu, Ying-Chang Liang, and Shaoqian Li are with the National Key Lab of Science and Technology on Communications, University of Electronic Science and Technology of China, Chengdu, China, emails: linzhang1913@gmail.com, gdngzhao@gmail.com, di$\_$di$\_$zhou@163.com, wugang99@uestc.edu.cn, liangyc@ieee.org, and lsq@uestc.edu.cn; Guodong Zhao is also with Department of Electrical and Computer Engineering, Lehigh University, Bethlehem, PA, USA.}
}
\maketitle

\thispagestyle{empty}

\begin{abstract}
  We study the coexistence problem in a two-tier heterogeneous network (HetNet) with cognitive small cells. In particular, we consider an underlay HetNet, where the cognitive small base station (C-SBS) is allowed to use the frequency bands of the macro cell with an access probability (AP) as long as the C-SBS satisfies a preset interference probability (IP) constraint at macro users (MUs). To enhance the AP (or transmission opportunity) of the C-SBS, we propose a learning-based algorithm for the C-SBS and exploit the distance information between the macro base station (MBS) and MUs. Generally, the signal from the MBS to a specific MU contains the distance information between the MBS to the MU. We enable the C-SBS to analyze the MBS signal on a target frequency band, and learn the distance information between the MBS and the corresponding MU. With the learnt distance information, we calculate the upper bound of the probability that the C-SBS may interfere with the MU, and design an AP with a closed-form expression under the IP constraint. Numerical results indicate that the proposed algorithm outperforms the existing methods up to $60\%$ AP (or transmission opportunity).
\end{abstract}

\begin{keywords}
Access probability, cognitive small cells, interference probability, learning, underlay HetNet.
\end{keywords}

\section{Introduction}

Recent years have witnessed an explosive growth of the wireless data traffic driven by mobile devices such as smart phones and tablets \cite{5G, Lin_caching}. To accommodate the tremorous data traffic, mobile operators are pushed to increase the network capacity. However, the existing cellular architecture is designed to provide wide area coverage of users and has limited network capacity. To deal with this issue, the concept of \emph{heterogeneous network} (HetNet) is proposed by embedding the conventional macro cell with multiple small cells \cite{HetNet_1, HetNet_2, HetNet_3,Survey_Two_tire_1, Survey_Two_tire_2, Survey_Two_tire_3}. Within a HetNet, a \emph{macro base station} (MBS) is equipped in the macro cell to provide wide area coverage of users. Meanwhile, a \emph{small base station} (SBS) in each small cell is responsible for providing high data rate access for the users within its coverage. Consequently, the HetNet is able to provide both wide area coverage and high network capacity, and is a promising candidate for the next generation of wireless communications.

In the deployment of the HetNet, one of the major challenges is the coexistence issue between small cells and the macro cell. Among the previous works in \cite{Overlay_HetNet_1,Overlay_HetNet_2,Overlay_HetNet_3, Underlay_HetNet_1,Underlay_HetNet_2, Underlay_HetNet_3}, there are mainly two categories of coexistence schemes. The first one is overlay HetNet \cite{Overlay_HetNet_1,Overlay_HetNet_2,Overlay_HetNet_3}, in which frequency resource is divided into two non-overlapped groups. Then, small cells and the macro cell use orthogonal frequency bands to avoid the co-channel interference, but this is inefficient in terms of spectrum efficiency, especially under a sparse small cell deployment. The second one is underlay HetNet \cite{Underlay_HetNet_1,Underlay_HetNet_2, Underlay_HetNet_3}, in which small cells and the macro cell share the same frequency bands. Then, the underlay HetNet can provide higher spectrum efficiency compared with the overlay HetNet. Nevertheless, a centralized coordinator is expected to manage the co-channel interference between small cells and the macro cell, and raises extra cost.

To deal with the coexistence issue in the HetNet, cognitive techniques are introduced to small cells, namely, cognitive small cells, in which the SBS is equipped with the cognitive capability, namely, cognitive SBS (C-SBS). In particular, the C-SBS is able to learn the status of the frequency bands in the macro cell. Then, the C-SBS can determine whether to access a frequency band or not. There are many contributions on the HetNet with cognitive small cells, e.g., \cite{HetNet_CR_1, HetNet_CR_2, HetNet_CR_Lin, HetNet_CR_3, HetNet_CR_4, HetNet_CR_5, HetNet_CR_Guodong}. Specifically, \cite{HetNet_CR_1, HetNet_CR_2, HetNet_CR_Lin} investigated the overlay HetNet and proposed a two-phase opportunistic spectrum access scheme. With this scheme, the C-SBS conducts spectrum sensing on a target frequency band in the first phase. If the frequency band is idle, the C-SBS accesses the frequency band in the second phase. Otherwise, the C-SBS accesses another frequency band or keeps silent. Although the C-SBS and the MBS share the same frequency bands, they use these frequency bands in orthogonal time slots. Then, the co-channel interference from the C-SBS to MUs can be avoided. However, the C-SBS may have rare access opportunities when idle frequency bands are limited, for instance, the traffic load of the macro cell is heavy. \cite{HetNet_CR_3, HetNet_CR_4, HetNet_CR_5, HetNet_CR_Guodong} investigated the underlay HetNet, where the C-SBS is allowed to share the same frequency bands with the MBS in the same time slots. In particular, \cite{HetNet_CR_3, HetNet_CR_4, HetNet_CR_5} adopted the independent stationary point process to model the locations of MUs and designed variable access policies to the busy frequency bands under different criterion, e.g., outage probability, throughput, and spectrum efficiency. \cite{HetNet_CR_Guodong} developed algorithms to detect whether a specific MU is inside the C-SBS coverage and designed an \emph{access probability} (AP) to the frequency band occupied by the MU under an \emph{interference probability} (IP) constraint.

From \cite{HetNet_CR_3, HetNet_CR_4, HetNet_CR_5, HetNet_CR_Guodong}, we notice that the instantaneous location information of a MU is crucial for the C-SBS to maximize the AP (or transmission opportunity) in the underlay HetNet. In particular, with the instantaneous location information, the C-SBS can access the frequency band when the MU is outside the C-SBS coverage, or the C-SBS can access the frequency band with an AP to satisfy the IP constraint when the MU is inside the C-SBS coverage. On the contrary, without the instantaneous location information, the C-SBS has to consider the worst case and conservatively access the frequency band all the time. This inevitably compromises the AP (or transmission opportunity) of the C-SBS. However, the instantaneous location information of the MU is only available at the MBS and the MU, and is unknown at the C-SBS. Thus, it is challenging for the C-SBS to obtain the instantaneous location information of MUs and maximize the AP (or transmission opportunity).

One straightforward alternative is to utilize the stochastic location information of MUs in \cite{HetNet_CR_3, HetNet_CR_4, HetNet_CR_5} and another approach is to utilize partial instantaneous location information of MUs in   \cite{HetNet_CR_Guodong}. However, both schemes provide limited AP (or transmission opportunity) for the C-SBS. In this paper, we propose a learning-based algorithm for the C-SBS and exploit the MBS-MU distance information to enhance the AP (or transmission opportunity). To our best knowledge, this is the first work to enable the C-SBS to learn the MBS-MU distance information and design the AP for the C-SBS to realize the underlay HetNet. To highlight our contributions, we summarize this paper as follows:

{$\bullet$} We propose a learning-based method for the C-SBS and exploit the MBS-MU distance information to enhance the AP (transmission opportunity) in the underlay HetNet.

{$\bullet$} By enabling the C-SBS to analyze the MBS signal, we learn the MBS-MU distance information, i.e., the relation between the MBS-MU distance and the received \emph{signal-to-noise ratio} (SNR) of the MBS signal at the C-SBS.

{$\bullet$} Based on the learnt MBS-MU distance information, we calculate the upper bound of the probability that the C-SBS may interfere with the MU. Then, we design an AP for the C-SBS with a closed-form expression under the IP constraint at the MU.

{$\bullet$} We numerically verify that the proposed algorithm outperforms the existing methods up to $60\%$ AP (or transmission opportunity).

The rest of the paper is organized as follows: Section II describes the system model in this paper. Section III develops algorithms for the C-SBS to learn the MBS-MU distance information. In Section IV, we exploit the learnt MBS-MU distance information and design an AP for the C-SBS u                                                                           nder the IP constraint. Section V provides numerical results and compares the proposed algorithm with the state of arts. Finally, Section VI concludes the paper.

\section{System Model}
We consider a HetNet model in Fig. {\ref{systemmodel}}, which consists of a macro cell with radius $R$ and a cognitive small cell with radius $r$. In particular, MUs are uniformly distributed in the macro cell and the MBS serves each MU with a certain frequency band. Meanwhile, the C-SBS is in the coverage of the MBS and is allowed to access the frequency band being occupied by a MU for coexistence. We denote the distance from the MBS to a MU and the C-SBS as $d_0$ and $d_1$, respectively. In what follows, we provide the channel model and the signal model, respectively.

         \begin{figure}[t!]
            \centering
            \includegraphics[scale=0.6]{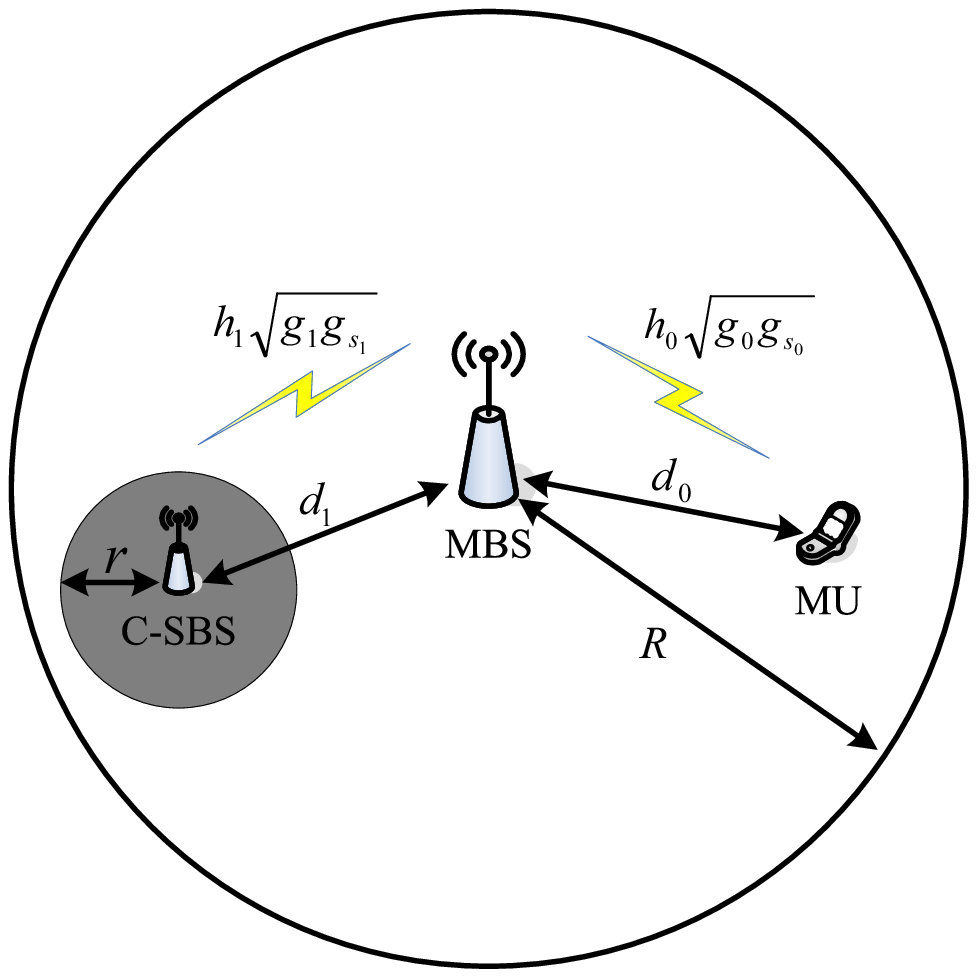}
            \caption{HetNet model, which consists of a macro cell with radius $R$ and a cognitive small cell with radius $r$. In particular, MUs are uniformly distributed in the macro and the MBS serves each MU with a certain frequency band. Meanwhile, the C-SBS is in the coverage of the MBS and is allowed to access the frequency band being occupied by a MU for coexistence. }
            \label{systemmodel}
        \end{figure}

\subsection{Channel Model}

We denote $h_0$ ($h_1$), $g_{s_0}$ ($g_{s_1}$), and $g_0$ ($g_1$) as the fading, the shadowing, and the path-loss coefficients between the MBS and MU (C-SBS), respectively. Then, the channel between the MBS and the MU (C-SBS) is ${h_0}\sqrt {{g_0}{g_{s_0}}}$ (${h_1}\sqrt {{g_1}{g_{s_1}}}$). In particular, $|h_q|$ ($q=0, \ 1$) follows a Rayleigh distribution with unit mean. $g_{s_q}$ ($q=0, \ 1$) follows a log-normal distribution with standard derivation $\sigma_s$. $g_q$ ($q=0, \ 1$) is determined by the path-loss model. If we adopt the path-loss model\footnote{Although we adopt the path loss model (\ref{Path_loss_model}) in this paper, our proposed algorithms can be used in other path loss models. }  \cite{3GPP_channel_model}
\begin{equation}
    P_l(d_q)=128 + 37.6\log_{10}(d_q), \ \ \text{for} \ \ d_q \geq
    \xi,
    \label{Path_loss_model}
\end{equation}
where $\xi=0.035$ km is the minimum distance between a transmitter and a receiver, $g_q$ can be expressed as
\begin{equation}
    g_q=10^{-12.8} d_q^{-3.76}, \ \ \text{for} \ \ d_q \geq \xi.
    \label{g_i}
\end{equation}

         \begin{figure}[t!]
            \centering
            \includegraphics[scale=0.42]{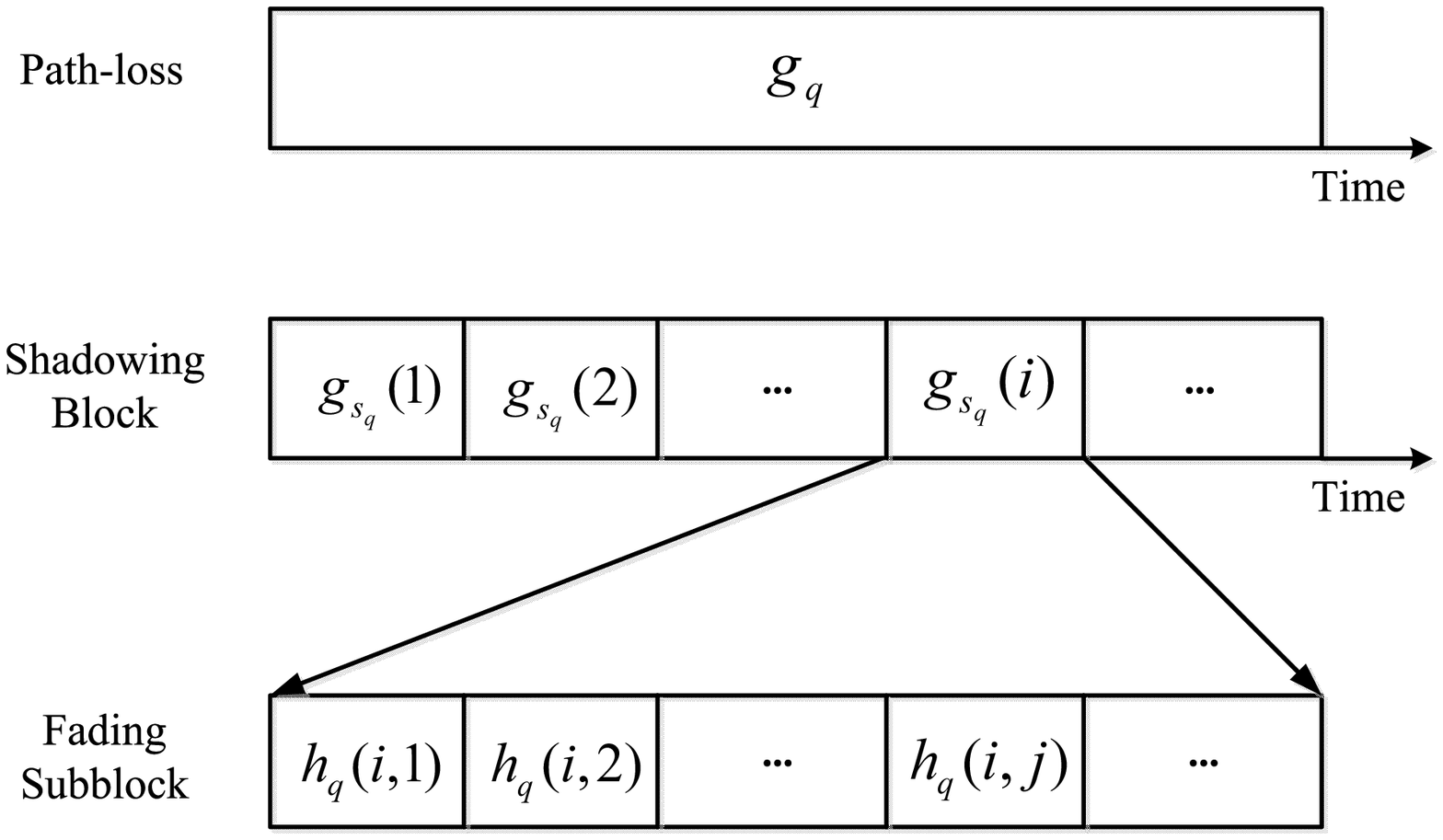}
            \caption{Channel model, where time axis is divided into blocks and each blocks consists of multiple subblocks. In particular, $g_q$ remains constant all the time with a given $d_q$, $g_{s_q}$ ($q=0, \ 1$) remains constant within each block ($i$) and varies independently among different blocks, and $h_q$ ($q=0, \ 1$) remains constant within each subblock ($i, j$) and varies among different subblocks.}
            \label{Channel_model}
        \end{figure}
For illustrations, we provide the channel model in Fig. \ref{Channel_model}, where time axis is divided into blocks and each block consists of multiple subblocks. In particular, $g_q$ remains constant all the time for a given $d_q$, $g_{s_q}$ ($q=0, \ 1$) remains constant within each block ($i$) and varies independently among different blocks, and $h_q$ ($q=0, \ 1$) remains constant within each subblock ($i, j$) and varies among different subblocks.

\subsection{Signal Model}

\subsubsection{Signal model from the MBS to the MU}
If we denote $x_{0}$ as the MBS signal with unit power, i.e., $|x_{0}|^2=1$, and denote $p_0$ as the transmit power of the MBS, the received signal at the MU is
 \begin{equation}
 y_0(i,j)= h_0(i,j)\sqrt {g_0g_{s_0}(i)p_0(i,j)} x_0(i,j) + n_0(i,j),
 \label{y_p}
 \end{equation}
where $(i)$ denotes the index of the $i$th block, $(i,j)$ represents the index of the $j$th subblock in the $i$th block,  $n_0$ represents the AWGN at the MU with zero mean and variance $\sigma^2$. Then, the SNR of the received signal at the MU is
 \begin{equation}
\gamma_{0}(i,j)= \frac{|h_0(i,j)|^2g_0g_{s_0}(i)p_0(i,j)}{\sigma^2}.
 \label{gamma_p}
 \end{equation}

We further consider that the MBS and the MU adopt \emph{close loop
power control} (CLPC) to provide \emph{quality of service} (QoS) guaranteed wireless communication \cite{Underlay_HetNet_4,Rui}. That is, the MBS automatically adjusts its transmit power to meet a certain target SNR $\gamma_T$ at the MU. Then, the transmit power of the MBS is
\begin{equation}
p_0(i,j)= \frac{{\gamma _T\sigma^2}}{{|h_0(i,j)|^2{g_0g_{s_0}(i)}}}.
\label{p_0}
\end{equation}

\subsubsection{Signal model from the MBS to the C-SBS}
In the meantime, the received MBS signal at the C-SBS is
  \begin{equation}
y_1(i,j) = h_1(i,j)\sqrt {{g_1}{g_{s_1}(i)}{p_0(i,j)}} {x_{0}}(i,j) + n_1(i,j),
 \label{y_c}
 \end{equation}
where $n_1$ is the AWGN at the C-SBS with zero mean and variance $\sigma^2$. Then, the SNR of the received MBS signal at the C-SBS is
 \begin{equation}
\gamma_1(i,j) = \frac{{|h_1(i,j)|^2{g_1}{g_{s_1}(i)}{p_{0}(i,j)}}}{{\sigma^2}}.
 \label{gamma_c}
 \end{equation}

By Substituting (\ref{g_i}) and (\ref{p_0}) into (\ref{gamma_c}), $\gamma_1(i,j)$ in (\ref{gamma_c}) can be rewritten as
\begin{equation}
{\gamma _{1}}(i,j) = \frac{{{\gamma _T}{d_1^{-3.76}}}}{{{d_0^{-3.76}}}}\frac{g_{s_1}(i)}{g_{s_0}(i)} \frac{{|h_1(i,j)|^2}}{{|h_0(i,j)|^2}}.
\label{gamma_c_1}
\end{equation}

\section{MBS-MU Distance Information Learning}
In this section, we develop an algorithm for the C-SBS to learn the MBS-MU distance information. In principle, when the MBS is transmitting signals to the MU with a target SNR, the transmission is based on the MBS-MU distance. In particular, if the MBS-MU distance is small, the MBS is able to satisfy the target SNR with a small transmit power. This leads to a small measured SNR at the C-SBS. Otherwise, the MBS has to increase its transmit power to achieve the target SNR. This results in a large measured SNR at the C-SBS. Therefore, the measured SNR at the C-SBS contains the MBS-MU distance information. In what follows, we will first calculate the \emph{cumulative density function} (CDF) of the measured SNR, and then obtain the relation between the MBS-MU distance and the measured SNR.

\subsection{CDF of the SNR at the C-SBS}
Removing the time index of the SNR at the C-SBS in (\ref{gamma_c_1}) and rewriting the SNR in dB, we have
\begin{align}
{\gamma _{1,dB}} = {\gamma _{T,dB}} + 37.6{\log _{10}}\left( \frac{d_0}{d_1} \right)+ \theta _r+\theta _s,
 \label{gamma_c_dB}
\end{align}
where ${\theta _r} = 10{\log _{10}}\left( |h_1|^2/|h_0|^2 \right)$ and ${\theta _s} = 10{\log _{10}}\left( {{g_{s_1}}} \right) - 10{\log _{10}}\left( {{g_{s_0}}} \right)$.

Then, the CDF of $\gamma_{1, dB}$ can be expressed as
\begin{align}
\label{gamma_c_dB_CDF_1}
\nonumber
& F_{\Gamma _{1,dB}}(\gamma_{1, dB}) \\ \nonumber
= &\Pr\left\{{\gamma _{T,dB}} + 37.6{\log _{10}}\left( \frac{d_0}{d_1} \right)+ {\theta _r} + {\theta _s}\leq \gamma_{1, dB}\right\}\\
= &\Pr\left\{{\theta _r}\! + {\theta _s}\leq \gamma_{1, dB}\!-{\gamma _{T,dB}}\!- 37.6{\log _{10}}\left( \frac{d_0}{d_1} \right)\right\},
\end{align}
which is related to the distributions of both $\theta _r$ and $\theta _s$. In the following, we first calculate the \emph{probability density function} (PDF) of both $\theta _r$ and $\theta _s$, and then obtain the CDF of $\gamma_{1, dB}$ with (\ref{gamma_c_dB_CDF_1}).

\subsubsection{PDF of $\theta _r$} Since both $|h_0|$ and $|h_1|$ follow a Rayleigh distribution with unit mean, the CDF of $\phi=|h_1|^2/|h_0|^2$ is \cite{Joint_PDF}
\begin{align}
F_{\Phi}(\phi) = \frac{\phi}{1 + \phi}.
\label{phi_cdf}
\end{align}

Then, the CDF of $\theta _r=10\log_{10}(\phi)$ is
\begin{align} \nonumber
F_{\Theta _r}(\theta _r) =& \Pr\left\{10\log_{10}(\phi) \leq \theta _r \right\}\\ \nonumber
=& \Pr\left\{\phi \leq 10^{\frac{\theta _r}{10}} \right\}\\ \nonumber
=& F_{\Phi}\left( 10^{\frac{\theta _r}{10}}\right)\\
=& \frac{1}{1 + 10^{-\frac{\theta _r}{10}}}.
\label{Theta_r_cdf}
\end{align}


Taking the derivation of $F_{\Theta _r}(\theta _r)$, we have the PDF of $\theta _r$ as
\begin{equation}\label{a5}
{f_{{\Theta _r}}}\left( {{\theta _r}} \right) = \frac{{\ln 10 \cdot {{10}^{ - \frac{{{\theta _r}}}{{10}}}}}}{{10{{\left( {1 + {{10}^{ - \frac{{{\theta _r}}}{{10}}}}} \right)}^2}}}.
\end{equation}

\subsubsection{PDF of $\theta _s$} Recall that both $g_{s_0}$ and $g_{s_1}$ follow a log-normal distribution with standard derivation $\sigma_s$. Then, it is straightforward to obtain that $\theta_s$ follows a normal distribution with zero means and variance $2\sigma_s^2$. Thus, the PDF of $\theta_s$ is
\begin{equation}
\label{Theta_2_pdf}
{f_{{\Theta _s}}}\left( {{\theta _s}} \right) = \frac{1}{{\sqrt {4\pi   \sigma _s^2 } }}{e^{ - \frac{{{{ {{\theta _s}} }^2}}}{{4 {\sigma _s^2}}}}}.
\end{equation}

Based on (\ref{a5}) and (\ref{Theta_2_pdf}), the CDF of $\gamma_{1, dB}$ in (\ref{gamma_c_dB_CDF_1}) can be calculated as
\begin{align}
\label{gamma_c_dB_CDF_2}
\nonumber
 F_{\Gamma _{1,dB}}(\gamma_{1, dB})=& \Pr\left\{{\theta _r} + {\theta _s}\leq m(\gamma_{1, dB})\right\}\\ \nonumber
=&\int_{ - \infty }^\infty  \int_{ - \infty }^{m(\gamma_{1, dB}) - {\theta _s}}\!\!f_{\Theta _s}\left(\theta _s \right) \!\!f_{\Theta _r}\left( \theta _r \right)d\theta _rd\theta _s\\
\!=\!&\int_{ - \infty }^\infty \!\! f_{\Theta _s}\left(\theta _s \right) F_{\Theta _r}\left(m(\gamma_{1, dB})\! - \! {\theta _s} \!\right)d\theta _s,
\end{align}
where $m(\gamma_{1, dB})=\gamma_{1, dB}-{\gamma _{T,dB}}- 37.6{\log _{10}}\left( \frac{d_0}{d_1}\right)$.

\subsection{Relation between $d_0$ and $\gamma_{1,dB}$}

To begin with, we provide the following definition.

\begin{definition} For a random variable $X$ with CDF $F_X(x)$, $x\in \mathbb{R}$, if $x_{\frac{1}{2}}$ satisfies both $F_X(x_{\frac{1}{2}})= \Pr\{X\leq x_{\frac{1}{2}}\}=\frac{1}{2}$ and $1-F_X(x_{\frac{1}{2}})= \Pr\{X\geq x_{\frac{1}{2}}\}=\frac{1}{2}$, $x_{\frac{1}{2}}$ is defined as the median of the random variable $X$.
\end{definition}

If we denote the median of the random variable $\gamma_{1,dB}$ as $\gamma_{1,dB,\frac{1}{2}}$, we have the following lemma.
\begin{lemma}
The median of the random variable $\gamma_{1,dB}$ is $\gamma_{1,dB,\frac{1}{2}}={\gamma _{T,dB}}+ 37.6{\log _{10}}\left( \frac{d_0}{d_1}\right)$.
\end{lemma}

\begin{proof}
The detailed proof of this Lemma is provided in Appendix A.
\end{proof}

From Definition 1 and Lemma 1, the probability that each sample of $\gamma_{1,dB}$ is smaller or larger than $\gamma_{1,dB,\frac{1}{2}}$ is equal to $\frac{1}{2}$. Suppose that the C-SBS observes MBS signals in $I$ blocks and measures $\gamma_{1, dB}$ of $J$ subblocks within each block. Then, the C-SBS can obtain $K=IJ$ independent samples of $\gamma_{1, dB}$, namely, $\gamma _{1, dB}(i,j)$ ($1\leq i\leq I$, $1\leq j\leq J$). By resorting these $K$ samples in an ascending order, we relabel these samples as $\bar{\gamma}_{1, dB}(k)$ ($1\leq k\leq K$). Then, we have the relation between the MBS-MU distance and the received SNR $\gamma_{1, dB}$ (or $\bar \gamma_{1, dB}$) at the C-SBS in the following Theorem.

\begin{theorem}
Define a function $f(x)=d_110^{\frac{x-\gamma _{T,dB}}{37.6}}$. Then, for a given distance $d_0^e$ from the MBS, the probability that the distance $d_0$ is larger than $d_0^e$ follows:

$\bullet$ If $d_0^{e} < f\left(\bar \gamma _{1,dB}(1)\right)$, we have
\begin{equation}
1- \left(1/2\right)^{K} < \Pr\left\{d_0^e < d_0\right\} <1.
\label{Th2_bound1}
\end{equation}

$\bullet$ If $d_0^{e} \geq f\left(\bar \gamma _{1,dB}(K)\right)$, we have
\begin{equation}
0 < \Pr\left\{d_0^e \leq d_0\right\}< \left(\frac{1}{2}\right)^{K}.
\label{Th2_bound2}
\end{equation}

$\bullet$ If $f\left(\bar \gamma _{1,dB}(k')\right) \leq d_0^e < f\left(\bar \gamma _{1,dB}(k'+1)\right)$, where $1\leq k' \leq K-1$, we have
\begin{align} \nonumber
&\left(\frac{1}{2}\right)^K\left(C_K^{k'+1}+C_K^{k'+2}+\cdots+C_K^{K}\right) \leq  \Pr\left\{d_0^e \leq d_0 \right\}\\
 \leq & \left(\frac{1}{2}\right)^K\left(C_K^{k'}+C_K^{k'+1}+\cdots+C_K^{K}\right).
 \label{Th2_bound3}
\end{align}

\end{theorem}

\begin{proof}
The detailed proof of this Theorem is provided in Appendix B.
\end{proof}

Theorem 1 provides both the upper bound and the lower bound of the probability that the MBS-MU distance is in a certain region. With this theorem, we will design the AP of the C-SBS and satisfy the IP constraint at the MU in the next section.

\section{Coexistence between the Cognitive Small Cell and the Macro cell}
In this section, we utilize the learnt MBS-MU distance information, i.e., relation between $d_0$ and $\bar \gamma_{1,dB}$, to design the AP for the coexistence between the cognitive small cell and the macro cell. We notice that the AP design may be different for distinct locations of the C-SBS within the macro cell. Thus, we will first classify the locations of the C-SBS into three scenarios. Then, we design the AP of the C-SBS in each scenario.

\subsection{Classification of Distinct Scenarios}

             \begin{figure}[t!]
            \centering
            \includegraphics[scale=0.45]{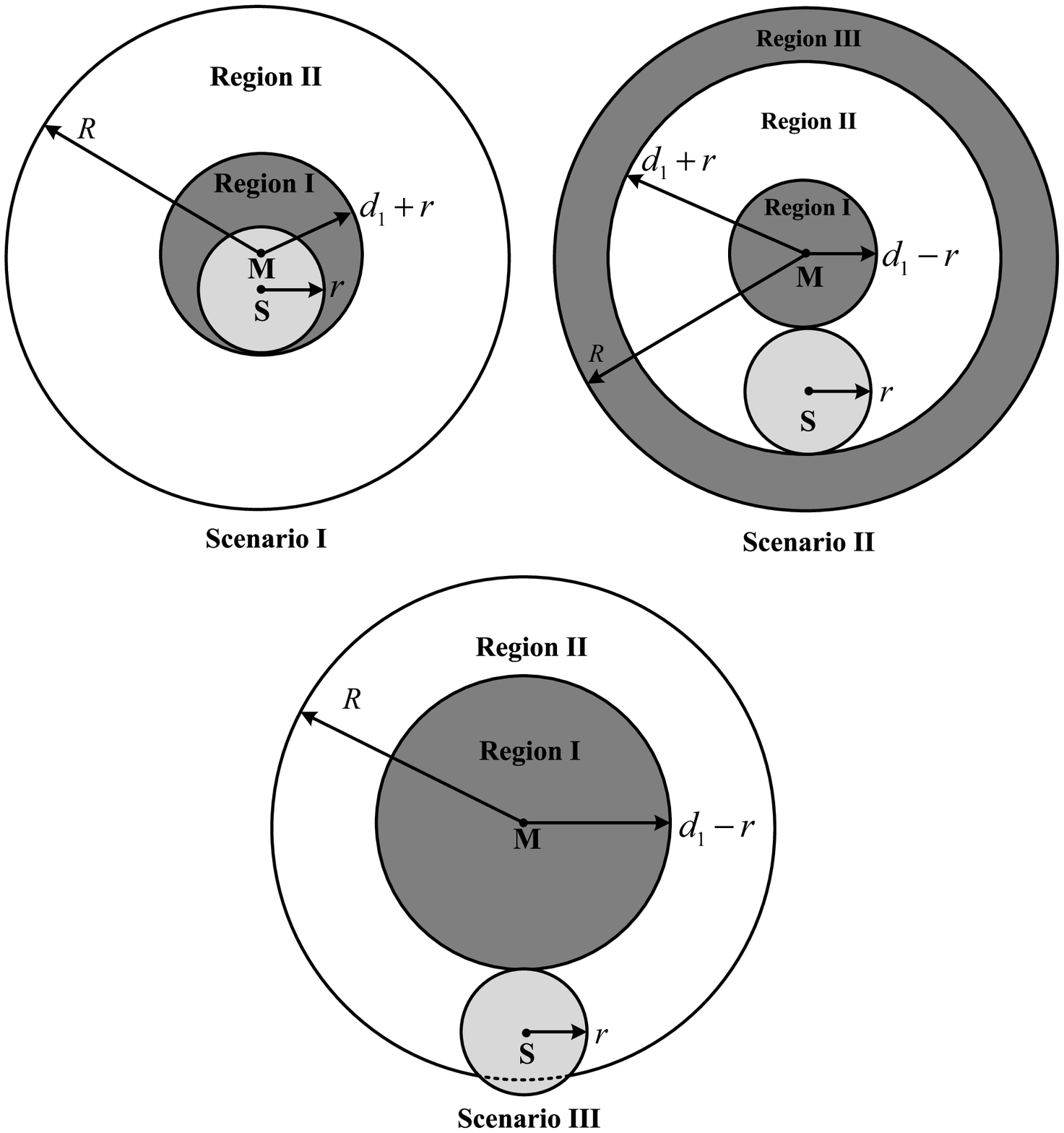}
            \caption{Illustrations of three scenarios, in which the C-SBS is in different locations within the macro cell.}
            \label{Position_interference_model}
        \end{figure}

Based on the MBS-C-SBS distance $d_1$, we consider three scenarios as shown in Fig. \ref{Position_interference_model}, where points ``M'' and ``S'' denote the MBS and the C-SBS, respectively.

{$\bullet$} In Scenario I, the MBS-C-SBS distance $d_1$ is small, i.e., $\xi \leq d_1\leq r+\xi$. Then, we divide the macro cell into two non-overlapped regions\footnote{Since the minimum distance between a transmitter and a receiver is $\xi>0$, we only consider the regions satisfying the minimum distance.}, i.e., Region I and Region II. In particular, Region I denotes the region where the distance from the MBS is between $\xi$ and $d_1+r$, and Region II denotes the region where the distance from the MBS is between $d_1+r$ and $R$.

{$\bullet$} In Scenario II, the MBS-C-SBS distance $d_1$ is medium, i.e., $r+\xi< d_1<R-r$. Then, we divide the macro cell into three non-overlapped regions, i.e., Region I, Region II, and Region III. In particular, Region I denotes the region where the distance from the MBS is between $\xi$ and $d_1-r$, and Region II denotes the region where the distance from the MBS is between $d_1-r$ and $d_1+r$, and Region III denotes the region where the distance from the MBS is between $d_1+r$ and $R$.

{$\bullet$} In Scenario III, the MBS-C-SBS distance $d_1$ is large, i.e., $R-r\leq d_1\leq R+r$. Then, we divide the macro cell into two non-overlapped regions, i.e., Region I and Region II. In particular, Region I denotes the region where the distance from the MBS is between $\xi$ and $d_1-r$, and Region II denotes the region where the distance from the MBS is between $d_1-r$ and $R$.

\subsection{AP Design in Scenario I}
In this part, we exploit the results in Theorem 1 to design the AP of the C-SBS and satisfy the IP constraint at the MU in Scenario I, where the macro cell is divided into two regions. In this scenario, the C-SBS may cause interference to the MU if the MU is located in Region I, i.e., $\xi \leq d_0\leq d_1+r$. According to the relation between the measured SNRs at the C-SBS and $\xi$, $\xi$ may satisfy
\begin{align}
0< \xi< f\left(\bar \gamma _{1,dB}(1)\right),
\label{I_k_0}
\end{align}
or
\begin{align}
f\left(\bar \gamma _{1,dB}(k_0)\right) \leq \xi < f\left(\bar \gamma _{1,dB}(k_0+1)\right),
\label{I_k_0_1}
\end{align}
for $1\leq k_0\leq K-1$, or
\begin{align}
f\left(\bar \gamma _{1,dB}(K)\right) \leq  \xi.
\label{I_k_0_2}
\end{align}

Meanwhile, according to the relation between the measured SNRs at the C-SBS and $d_1+r$, $d_1+r$ may satisfy
\begin{align}
0< d_1+r< f\left(\bar \gamma _{1,dB}(1)\right),
\label{I_k_2}
\end{align}
or
\begin{align}
f\left(\bar \gamma _{1,dB}(k_2)\right) \leq d_1+r < f\left(\bar \gamma _{1,dB}(k_2+1)\right),
\label{I_k_2_1}
\end{align}
for $1\leq k_2\leq K-1$, or
\begin{align}
f\left(\bar \gamma _{1,dB}(K)\right) \leq d_1+r.
\label{I_k_2_2}
\end{align}

Due to the increasing monotonicity of $f(x)$ and $\xi<d_1+r$, we have six cases of $\xi$ and $d_1+r$: Case I): $\xi$ satisfies (\ref{I_k_0}) and $d_1+r$ satisfies (\ref{I_k_2}); Case II): $\xi$ satisfies (\ref{I_k_0}) and $d_1+r$ satisfies (\ref{I_k_2_1}); Case III): $\xi$ satisfies (\ref{I_k_0}) and $d_1+r$ satisfies (\ref{I_k_2_2}); Case IV) $\xi$ satisfies (\ref{I_k_0_1}) and $d_1+r$ satisfies (\ref{I_k_2_1}); Case V) $\xi$ satisfies (\ref{I_k_0_1}) and $d_1+r$ satisfies (\ref{I_k_2_2}); Case VI) $\xi$ satisfies (\ref{I_k_0_2}) and $d_1+r$ satisfies (\ref{I_k_2_2}).

Next, we first calculate the probability $\rho_{\text{I}}$ that the MU is in Region I for each case and then design the corresponding AP. In the sequential, we denote $S_{\text{I}}=\pi(d_1+r)^2-\pi\xi^2$ as the area of Region I, and denote $S_{\text{c}}$ as the area of the interference region, where the MU may appear in the C-SBS coverage. From Appendix C, we have
 \begin{align}
S_{\text{c}}\!\!=\!\!\left\{
\begin{aligned}
& \!\! \pi r^2-\pi\xi^2,  & \!\! &\xi \leq d_1 < r-\xi,\\
& \!\! (\pi\!-\!\varphi_2) r^2\!-\!\varphi_1\xi^2\!+\!d_1\xi \sin \varphi_1,  & \!\! & r \!-\!\xi \!\leq \!d_1\! \leq \!r\!+\!\xi,
\end{aligned}
\right.
\label{S_c_SI}
\end{align}
 where $\varphi_1=\arccos \frac{d_1^2+\xi^2-r^2}{2d_1\xi}$, $\varphi_2=\arccos \frac{d_1^2+r^2-\xi^2}{2d_1r}$.

\subsubsection{AP Design in Case I}
In this case, $\xi$ satisfies (\ref{I_k_0}) and $d_1+r$ satisfies (\ref{I_k_2}). The probability $\rho_{\text{I}}$ that the MU is in Region I is
\begin{align} \nonumber
\rho_{\text{I}}=&\Pr\left\{\xi\leq d_0\leq d_1+r\right\}\\ \nonumber
=& \Pr\left\{d_0 \geq \xi \right\}\!-\!\Pr\left\{d_0 \geq d_1+r\right\} \\ \nonumber
\overset{\text{(I-a)}}{\leq} & 1- 1+\left(\frac{1}{2}\right)^K\\
=&\left(\frac{1}{2}\right)^K,
\label{I_I_bound1}
\end{align}
where we use (\ref{Th2_bound1}) in (I-a).

To control the IP constraint $\eta$ at the MU, the AP of the C-SBS needs to satisfy $0\leq \rho_{\text{AP}}\leq 1$ and $\rho_{\text{AP}}\rho_{\text{I}}\frac{S_{\text{c}}}{S_{\text{I}}} \leq \eta$, i.e., $\rho_{\text{AP}} \leq \min \left\{\frac{\eta S_{\text{I}}}{\rho_{\text{I}}S_{\text{c}}},1\right\}=\min \left\{\frac{\eta [\pi(d_1+r)^2-\pi\xi^2]}{\rho_{\text{I}}S_{\text{c}}},1\right\}$, which can be lower bounded by
\begin{align} \nonumber
&\min\left\{\frac{\eta [\pi(d_1+r)^2-\pi\xi^2]}{\rho_{\text{I}}S_{\text{c}}},1\right\}\\
\overset{\text{(I-b)}}{\geq}& \min \left\{\frac{\eta 2^K[\pi(d_1\!+\!r)^2-\pi\xi^2]}{S_{\text{c}}},1\right\},
\label{I_I_bound2_SI}
\end{align}
where we use (\ref{I_I_bound1}) in (I-b).

Thus, to protect the MU in this case, the maximum AP of the C-SBS is
\begin{align}
 \rho_{\text{AP}}=\min \left\{\frac{\eta 2^K[\pi(d_1+r)^2-\pi\xi^2]}{S_{\text{c}}},1\right\},
\label{I_I_AP}
\end{align}
where $S_{\text{c}}$ is shown in (\ref{S_c_SI}).

\subsubsection{AP Design in Case II}
In this case, $\xi$ satisfies (\ref{I_k_0}) and $d_1+r$ satisfies (\ref{I_k_2_1}). The probability $\rho_{\text{I}}$ that the MU is in Region I is
\begin{align} \nonumber
\rho_{\text{I}}=& \Pr\left\{d_0 \geq \xi \right\}-\Pr\left\{d_0 \geq d_1+r\right\} \\ \nonumber
\overset{\text{(I-c)}}{\leq} & 1- \left(\frac{1}{2}\right)^K\left(C_K^{k_2+1}+C_K^{k_2+2}+\cdots+C_K^{K}\right)\\
=&\left(\frac{1}{2}\right)^K\left(C_K^{0}+C_K^{1}+\cdots+C_K^{k_2}\right),
\label{I_II_bound1}
\end{align}
where we use (\ref{Th2_bound1}) and (\ref{Th2_bound3}) in (I-c).

To control the IP constraint $\eta$ at the MU, the AP of the C-SBS needs to satisfy $0\leq \rho_{\text{AP}}\leq 1$ and $\rho_{\text{AP}}\rho_{\text{I}}\frac{S_{\text{c}}}{S_{\text{I}}} \leq \eta$, i.e., $\rho_{\text{AP}} \leq \min \left\{\frac{\eta S_{\text{I}}}{\rho_{\text{I}}S_{\text{c}}},1\right\}=\min \left\{\frac{\eta ([\pi(d_1+r)^2-\pi\xi^2]}{\rho_{\text{I}}S_{\text{c}}},1\right\}$, which can be lower bounded by
\begin{align}  \nonumber
&\min \left\{\frac{\eta [\pi(d_1+r)^2-\pi\xi^2]}{\rho_{\text{I}}S_{\text{c}}},1\right\}\\
\overset{\text{(I-d)}}{\geq} &   \min \left\{\frac{\eta 2^K[\pi(d_1+r)^2-\pi\xi^2]}{(C_K^{0}+C_K^{1}+\cdots+C_K^{k_2})S_{\text{c}}},1\right\},
\label{I_II_bound2_SI}
\end{align}
where we use (\ref{I_II_bound1}) in (I-d).

Thus, to protect the MU in this case, the maximum AP of the C-SBS is
\begin{align}
 \rho_{\text{AP}}=\min \left\{\frac{\eta 2^K[\pi(d_1+r)^2-\pi\xi^2]}{(C_K^{0}+C_K^{1}+\cdots+C_K^{k_2})S_{\text{c}}},1\right\},
\label{I_II_AP}
\end{align}
where $S_{\text{c}}$ is shown in (\ref{S_c_SI}).

\subsubsection{AP Design in Case III}
In this case, $\xi$ satisfies (\ref{I_k_0}) and $d_1+r$ satisfies (\ref{I_k_2_2}). The probability $\rho_{\text{I}}$ that the MU is in Region II is
\begin{align}
\rho_{\text{I}}= \Pr\left\{d_0 \geq \xi \right\}-\Pr\left\{d_0\geq d_1+r\right\}
\overset{\text{(I-e)}}{\leq}  1,
\label{I_III_bound1}
\end{align}
where we use (\ref{Th2_bound1}) and (\ref{Th2_bound2}) in (I-e).

To control the IP constraint $\eta$ at the MU, the AP of the C-SBS needs to satisfy $0\leq \rho_{\text{AP}}\leq 1$ and $\rho_{\text{AP}}\rho_{\text{I}}\frac{S_{\text{c}}}{S_{\text{I}}} \leq \eta$, i.e., $\rho_{\text{AP}} \leq \min \left\{\frac{\eta S_{\text{I}}}{\rho_{\text{I}}S_{\text{c}}},1\right\}=\min \left\{\frac{\eta [\pi(d_1+r)^2-\pi\xi^2]}{\rho_{\text{I}}S_{\text{c}}},1\right\}$, which can be lower bounded by
\begin{align} \nonumber
&\min \left\{\frac{\eta [\pi(d_1+r)^2-\pi\xi^2]}{\rho_{\text{I}}S_{\text{c}}},1\right\}\\
\overset{\text{(I-f)}}{\geq} & \min \left\{ \frac{\eta[\pi(d_1+r)^2-\pi\xi^2]}{S_{\text{c}}},1\right\},
\label{I_III_bound2}
\end{align}
where we use (\ref{I_III_bound1}) in (I-f).

Thus, to protect the MU in this case, the maximum AP of the C-SBS is
\begin{align}
\rho_{\text{AP}}= \min \left\{ \frac{\eta[\pi(d_1+r)^2-\pi\xi^2]}{S_{\text{c}}},1\right\},
\label{I_III_AP}
\end{align}
where $S_{\text{c}}$ is shown in (\ref{S_c_SI}).

\subsubsection{AP Design in Case IV}
In this case, $\xi$ satisfies (\ref{I_k_0_1}) and $d_1+r$ satisfies (\ref{I_k_2_1}). The probability $\rho_{\text{I}}$ that the MU is in Region II is
\begin{align} \nonumber
\rho_{\text{I}}=& \Pr\left\{d_0\geq \xi \right\}-\Pr\left\{d_0\geq d_1+r\right\} \\  \nonumber
\overset{\text{(I-g)}}{\leq} & \left(\frac{1}{2}\right)^K\left(C_K^{k_0}+C_K^{k_0+1}+\cdots+C_K^{K}\right)\\ \nonumber
&-\left(\frac{1}{2}\right)^K\left(C_K^{k_2+1}+C_K^{k_2+2}+\cdots+C_K^{K}\right)\\
=&\left(\frac{1}{2}\right)^K\left(C_K^{k_0}+C_K^{k_0+1}+\cdots+C_K^{k_2}\right),
\label{I_IV_bound1}
\end{align}
where we use (\ref{Th2_bound3}) in (I-g).

To control the IP constraint $\eta$ at the MU, the AP of the C-SBS needs to satisfy $0\leq \rho_{\text{AP}}\leq 1$ and $\rho_{\text{AP}}\rho_{\text{I}}\frac{S_{\text{c}}}{S_{\text{I}}} \leq \eta$, i.e., $\rho_{\text{AP}} \leq \min \left\{\frac{\eta S_{\text{I}}}{\rho_{\text{I}}S_{\text{c}}},1\right\}=\min \left\{\frac{\eta [\pi(d_1+r)^2-\pi\xi^2]}{\rho_{\text{I}}S_{\text{c}}},1\right\}$, which can be lower bounded by
\begin{align} \nonumber
&\min \left\{\frac{\eta [\pi(d_1+r)^2-\pi\xi^2]}{\rho_{\text{I}}S_{\text{c}}},1\right\}\\
\overset{(\text{I-h})}{\geq} & \min \left\{\frac{\eta 2^K(\pi[\pi(d_1+r)^2-\pi\xi^2] }{(C_K^{k_0}+C_K^{k_0+1}+\cdots+C_K^{k_2})S_{\text{c}}},1\right\},
\label{I_IV_bound2}
\end{align}
where we use (\ref{I_IV_bound1}) in (I-h).

Thus, to protect the MU in this case, the maximum AP of the C-SBS is
\begin{align}
\rho_{\text{AP}}= \min \left\{\frac{\eta 2^K[\pi(d_1+r)^2-\pi\xi^2] }{(C_K^{k_0}+C_K^{k_0+1}+\cdots+C_K^{k_2})S_{\text{c}}},1\right\},
\label{I_IV_AP}
\end{align}
where $S_{\text{c}}$ is shown in (\ref{S_c_SI}).

\subsubsection{AP Design in Case V}
In this case, $\xi$ satisfies (\ref{I_k_0_1}) and $d_1+r$ satisfies (\ref{I_k_2_2}). The probability $\rho_{\text{I}}$ that the MU is in Region I is
\begin{align} \nonumber
\rho_{\text{I}}=&\Pr\left\{d_0\geq \xi \right\}-\Pr\left\{d_0 \geq d_1+r\right\}\\
\overset{\text{(I-i)}}{\leq} &  \left(\frac{1}{2}\right)^K\left(C_K^{k_0}+C_K^{k_0+1}+\cdots+C_K^{K}\right),
\label{I_V_bound1}
\end{align}
where we use (\ref{Th2_bound2}) and (\ref{Th2_bound3}) in (I-i).

To control the IP constraint $\eta$ at the MU, the AP of the C-SBS needs to satisfy $0\leq \rho_{\text{AP}}\leq 1$ and $\rho_{\text{AP}}\rho_{\text{I}}\frac{S_{\text{c}}}{S_{\text{I}}} \leq \eta$, i.e., $\rho_{\text{AP}} \leq \min \left\{\frac{\eta S_{\text{I}}}{\rho_{\text{I}}S_{\text{c}}},1\right\}=\min \left\{\frac{\eta [\pi(d_1+r)^2-\pi\xi^2]}{\rho_{\text{I}}S_{\text{c}}},1\right\}$, which can be lower bounded by
\begin{align} \nonumber
& \min \left\{\frac{\eta [\pi(d_1+r)^2-\pi\xi^2]}{\rho_{\text{I}}S_{\text{c}}},1\right\}\\
\overset{(\text{I-j})}{\geq} & \min \left\{\frac{\eta[\pi(d_1+r)^2-\pi\xi^2]}{(C_K^{k_0}+C_K^{k_0+1}+\cdots+C_K^{K})S_{\text{c}}},1\right\},
\label{I_IV_bound2}
\end{align}
where we use (\ref{I_V_bound1}) in (I-j).

Thus, to protect the MU in this case, the maximum AP of the C-SBS is
\begin{align}
\rho_{\text{AP}}= \min \left\{\frac{\eta[\pi(d_1+r)^2-\pi\xi^2]}{(C_K^{k_0}+C_K^{k_0+1}+\cdots+C_K^{K})S_{\text{c}}},1\right\},
\label{I_V_AP}
\end{align}
where $S_{\text{c}}$ is shown in (\ref{S_c_SI}).

\subsubsection{AP Design in Case VI}
In this case, $\xi$ satisfies (\ref{I_k_0_2}) and $d_1+r$ satisfies (\ref{I_k_2_2}). The probability $\rho_{\text{I}}$ that the MU is in Region I is
\begin{align}
\rho_{\text{I}}\!=\!\Pr\left\{d_0\!\geq \!\xi\right\}\!-\!\Pr\left\{d_0\!\geq\! d_1+r\right\}
\overset{\text{(I-k)}}{\leq}  \left(\frac{1}{2}\right)^K,
\label{I_VI_bound1}
\end{align}
where we use (\ref{Th2_bound2}) in (I-k).

To control the IP constraint $\eta$ at the MU, the AP of the C-SBS needs to satisfy $0\leq \rho_{\text{AP}}\leq 1$ and $\rho_{\text{AP}}\rho_{\text{I}}\frac{S_{\text{c}}}{S_{\text{I}}} \leq \eta$, i.e., $\rho_{\text{AP}} \leq \min \left\{\frac{\eta S_{\text{I}}}{\rho_{\text{I}}S_{\text{c}}},1\right\}=\min \left\{\frac{\eta [\pi(d_1+r)^2-\pi\xi^2]}{\rho_{\text{I}}S_{\text{c}}},1\right\}$, which can be lower bounded by
\begin{align} \nonumber
&\min \left\{\frac{\eta[\pi(d_1+r)^2-\pi\xi^2]}{\rho_{\text{I}}S_{\text{c}}},1\right\}\\
\overset{\text{(I-l)}}{\geq} & \min \left\{\frac{\eta2^K[\pi(d_1+r)^2-\pi\xi^2]}{S_{\text{c}}},1\right\},
\label{I_VI_bound2}
\end{align}
where we use (\ref{I_VI_bound1}) in (I-l).

Thus, to protect the MU in this case, the maximum AP of the C-SBS is
\begin{align}
\rho_{\text{AP}}=\min \left\{\frac{\eta2^K[\pi(d_1+r)^2-\pi\xi^2]}{S_{\text{c}}},1\right\},
\label{I_VI_AP}
\end{align}
where $S_{\text{c}}$ is shown in (\ref{S_c_SI}).

\subsection{AP design in Scenario II}
In this part, we exploit the results in Theorem 1 to design the AP of the C-SBS to satisfy the IP constraint at the MU in Scenario II, where the macro cell is divided into three regions. In this scenario, the C-SBS may cause interference to the MU if the MU is located in Region II, i.e., $d_1-r \leq d_0\leq d_1+r$. According to relation between the measured SNRs at the C-SBS and $d_1-r$, $d_1-r$ may satisfy
\begin{align}
0< d_1-r < f\left(\bar \gamma _{1,dB}(1)\right),
\label{II_k_1}
\end{align}
or
\begin{align}
f\left(\bar \gamma _{1,dB}(k_1)\right) \leq d_1-r < f\left(\bar \gamma _{1,dB}(k_1+1)\right),
\label{II_k_1_1}
\end{align}
for $1\leq k_1 \leq K-1$, or
\begin{align}
f\left(\bar \gamma _{1,dB}(K)\right) \leq d_1-r.
\label{II_k_1_2}
\end{align}

Meanwhile, according to the relation between the measured SNRs at the C-SBS and $d_1+r$, $d_1+r$ may satisfy (\ref{I_k_2}), or (\ref{I_k_2_1}), or (\ref{I_k_2_2}).

Due to the increasing monotonicity of $f(x)$ and $\xi<d_1+r$, we have six cases of $d_1-r$ and $d_1+r$: Case I): $d_1-r$ satisfies (\ref{II_k_1}) and $d_1+r$ satisfies (\ref{I_k_2}); Case II): $d_1-r$ satisfies (\ref{II_k_1}) and $d_1+r$ satisfies (\ref{I_k_2_1}); Case III): $d_1-r$ satisfies (\ref{II_k_1}) and $d_1+r$ satisfies (\ref{I_k_2_2}); Case IV) $d_1-r$ satisfies (\ref{II_k_1_1}) and $d_1+r$ satisfies (\ref{I_k_2_1}); Case V) $d_1-r$ satisfies (\ref{II_k_1_1}) and $d_1+r$ satisfies (\ref{I_k_2_2}); Case VI) $d_1-r$ satisfies (\ref{II_k_1_2}) and $d_1+r$ satisfies (\ref{I_k_2_2}). Since these six cases in Scenario II are similar to those in Scenario I, we analyze each case and design the corresponding AP in Appendix D.

\subsection{AP Design in Scenario III}
In this part, we exploit the results in Theorem 1 to design the AP of the C-SBS to satisfy the IP constraint at the MU in Scenario III, where the macro cell is divided into two regions. In this scenario, the C-SBS may cause interference to the MU if the MU is located in Region II, i.e., $d_1-r \leq d_0 \leq R$. According to relation between the measured SNRs at the C-SBS and $d_1-r$, $d_1-r$ may satisfy (\ref{II_k_1}), or (\ref{II_k_1_1}), or (\ref{II_k_1_2}).

Meanwhile, according to the relation between the measured SNRs at the C-SBS and $R$, $R$ may satisfy
 \begin{align}
0< R < f\left(\bar \gamma _{1,dB}(1)\right),
\label{III_k_3}
\end{align}
or
\begin{align}
f\left(\bar \gamma _{1,dB}(k_1)\right) \leq R < f\left(\bar \gamma _{1,dB}(k_1+1)\right),
\label{III_k_3_1}
\end{align}
for $1\leq k_3 \leq K-1$, or
\begin{align}
f\left(\bar \gamma _{1,dB}(K)\right) \leq R.
\label{III_k_3_2}
\end{align}

Due to the increasing monotonicity of $f(x)$ and $d_1-r<R$, we have six cases of $d_1-r$ and $R$: Case I): $d_1-r$ satisfies (\ref{II_k_1}) and $R$ satisfies (\ref{III_k_3}); Case II): $d_1-r$ satisfies (\ref{II_k_1}) and $R$ satisfies (\ref{III_k_3_1}); Case III): $d_1-r$ satisfies (\ref{II_k_1}) and $R$ satisfies (\ref{III_k_3_2}); Case IV) $d_1-r$ satisfies (\ref{II_k_1_1}) and $R$ satisfies (\ref{III_k_3_1}); Case V) $d_1-r$ satisfies (\ref{II_k_1_1}) and $R$ satisfies (\ref{III_k_3_2}); Case VI) $d_1-r$ satisfies (\ref{II_k_1_2}) and $R$ satisfies (\ref{III_k_3_2}). Since these six cases in Scenario III are similar to those in Scenario I, we analyze each case and design the corresponding AP in Appendix E.

\section{Numerical results}
In this section, we provide numerical results to demonstrate the performance of the proposed algorithm. To demonstrate the advantages of the proposed algorithm, we compare our results with the algorithm in \cite{HetNet_CR_Guodong} and the algorithm based on the Statistic Location Information of MUs in \cite{HetNet_CR_3, HetNet_CR_4, HetNet_CR_5} (referred to SLI algorithm hereafter). In the simulation, we adopt the system model in Section II, where the radius of the MBS coverage is $R=0.5$ km, the radius of the C-SBS coverage is $r=0.1$ km, and $10^5$ MUs are uniformly distributed in the MBS coverage. Furthermore, we assume the power of the AWGN $\sigma^2=-114$ dBm, the target SNR of the MU is $\gamma_{T,dB}=20$, the IP constraint at the MU is $\eta=0.01$, the number of blocks $I=50$ unless otherwise specified, and the number of subblocks within each block is $J=1$.

             \begin{figure}[t!]
            \centering
            \includegraphics[scale=0.5]{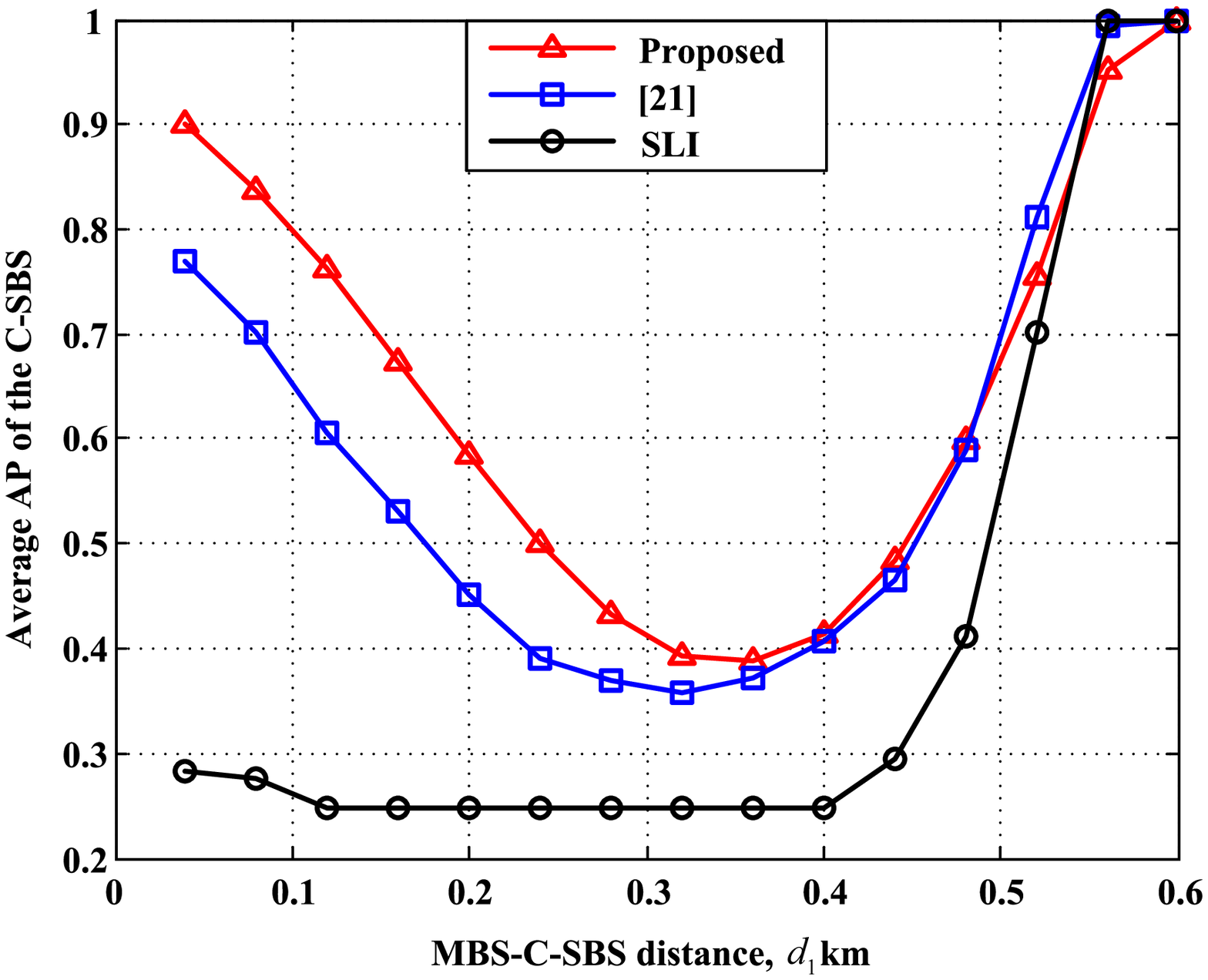}
            \caption{The comparison of the APs with different algorithms.}
            \label{AP_d1}
        \end{figure}
Fig. \ref{AP_d1} compares the AP of the proposed algorithm with the algorithm in \cite{HetNet_CR_Guodong} and the SLI. From this figure, we observe that the curve of the AP with the proposed algorithm is a ``U'' shape. In particular, the AP of the proposed algorithm first decreases from around $0.9$ to around $0.4$ as the MBS-C-SBS distance $d_1$ grows from $0.04$ km to $0.36$ km, and then increases from around $0.4$ to around $1$ as $d_1$ grows from $0.36$ km to $0.6$ km. The reason is that, when $d_1$ is small, the area of the interference region increases as $d_1$ grows. This reduces the AP of the C-SBS. When $d_1$ is large, the area of the interference region reduces as $d_1$ grows. This increases the AP of the C-SBS. Besides, the proposed algorithm outperforms both the algorithm in \cite{HetNet_CR_Guodong} and the SLI, when the C-SBS is in the MBS coverage, i.e., $d_1\leq R=0.5$ km. Specifically, the proposed algorithm is able to improve the AP by up to $15\%$ compared with the algorithm in \cite{HetNet_CR_Guodong} and improve the AP by up to $60\%$ compared with the SLI algorithm. This figure also indicates that the C-SBS can obtain a high AP when the C-SBS is deployed close to or far away from the MBS.

           \begin{figure}[t!]
            \centering
            \includegraphics[scale=0.5]{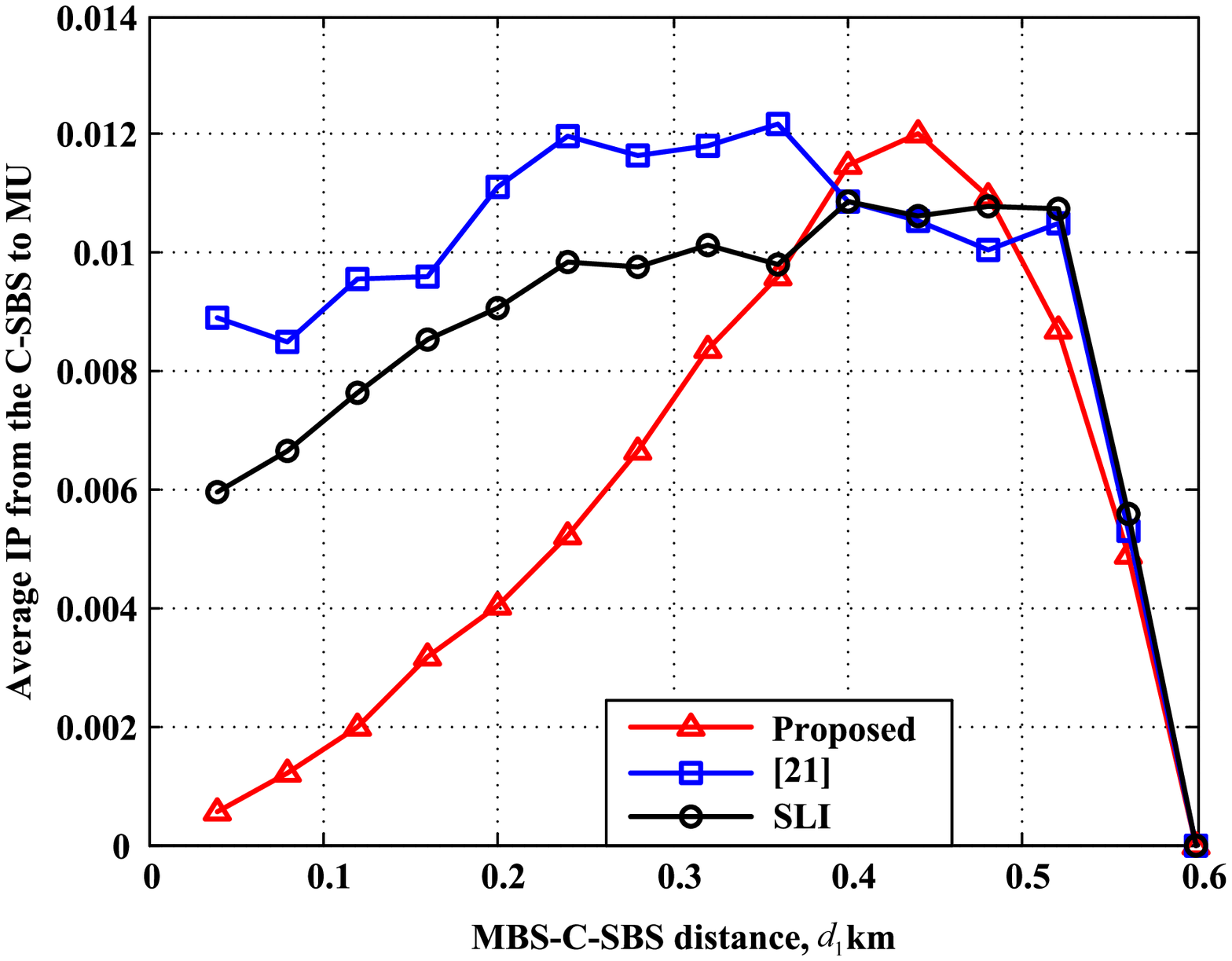}
            \caption{IP from the C-SBS to the MU with different algorithms corresponding to the AP in Fig. \ref{AP_d1}.}
            \label{IP_d1}
        \end{figure}

Fig. \ref{IP_d1} provides the IP of different algorithms corresponding to the AP in Fig. \ref{AP_d1}. In general, the IP with the proposed algorithm is below the IP constraint $\eta=0.01$ except for $0.37 \leq d_1\leq 0.5$ km, where the maximum IP is around $0.012$. In fact, the accuracy of the designed AP is affected by the AWGN. With a finite number of samples $K$, the designed AP is not accurate. Thus, there is a chance that the IP constraint $\eta=0.01$ at the MU is violated with the designed AP. Besides, we observe that the IP with the algorithm in \cite{HetNet_CR_Guodong} is above the IP constraint $\eta=0.01$ for $0.17 \leq d_1\leq 0.52$ km. And the IP with the SLI algorithm is above the IP constraint $\eta=0.01$ for $0.24 \leq d_1\leq 0.52$ km. Thus, the C-SBS with the proposed algorithm can be deployed in a wider region than the C-SBS with the algorithm in \cite{HetNet_CR_Guodong} and the SLI algorithm.

           \begin{figure}[t!]
            \centering
            \includegraphics[scale=0.5]{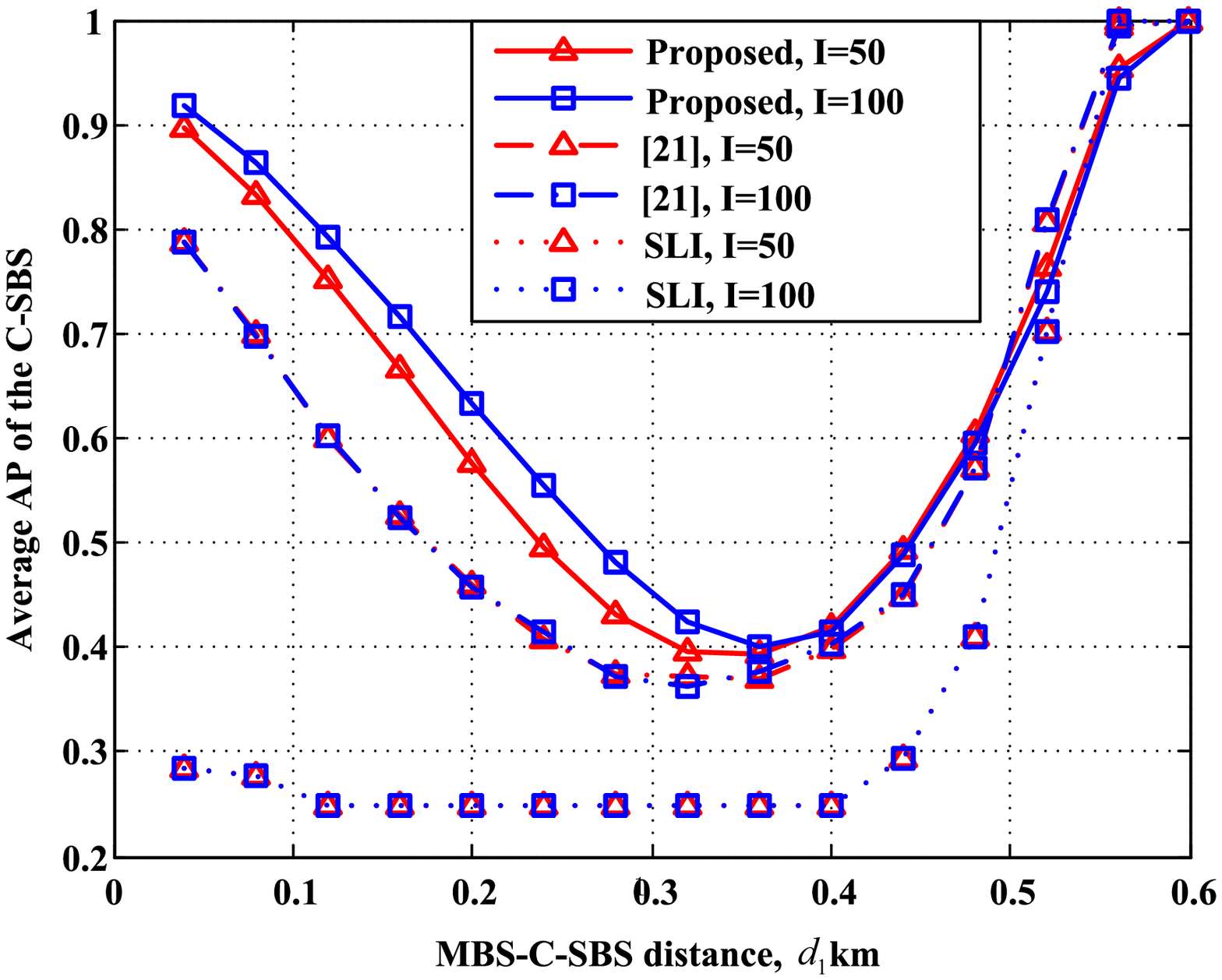}
            \caption{AP with different number of blocks $I$.}
            \label{AP_blocks}
        \end{figure}
Fig. \ref{AP_blocks} investigates the impact of the number of blocks $I$ on the AP of different algorithms. With the proposed algorithm, the AP increases as $I$ grows from $50$ to $100$ for small $d_1$, i.e., $d_1\leq 0.36$ km, and remains constant as $I$ grows from $50$ to $100$ for large $d_1$, i.e., $d_1> 0.36$ km. The reason is that, each measured SNR at the C-SBS contains both the information of $d_1$ and the noise. When $d_1$ is small, the measured SNR is large and contains more information of $d_1$ than the noise. Then, the C-SBS can learn more accurate information of $d_1$ with more blocks, and obtain a larger AP. When $d_1$ is large, the measured SNR is small and contains more noise than the information of $d_1$. Then, the increase of blocks has little impact on the AP. On the other hand, the AP always remains constant as $I$ grows from $50$ to $100$ with the algorithm in \cite{HetNet_CR_Guodong} and the SLI algorithm.

           \begin{figure}[t!]
            \centering
            \includegraphics[scale=0.5]{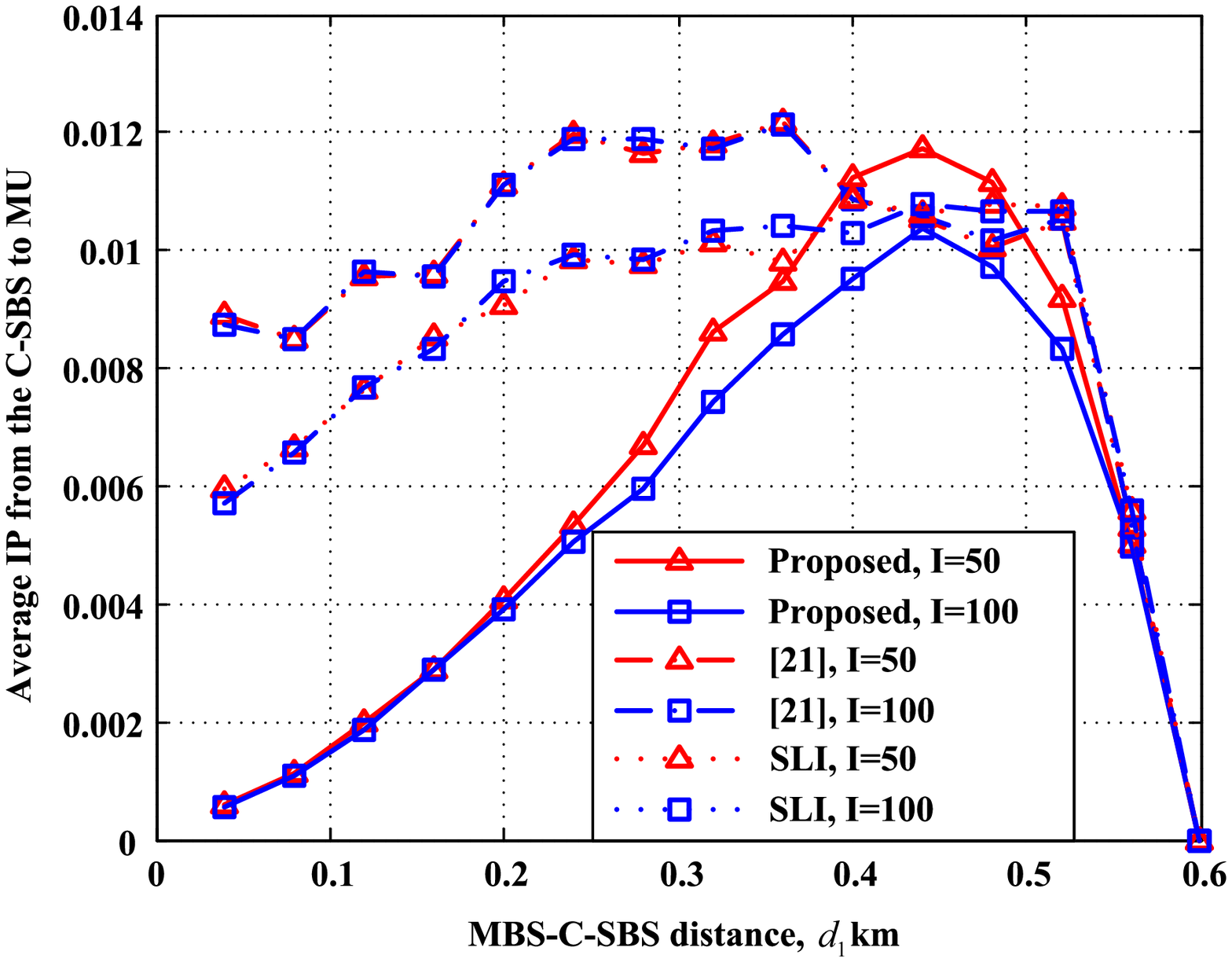}
            \caption{IP with different number of blocks $I$.}
            \label{IP_blocks}
        \end{figure}

Fig. \ref{IP_blocks} investigates the impact of the number of blocks $I$ on the IP of different algorithms. With the proposed algorithm, the IP decreases as $I$ grows. In particular, the IP is generally smaller than the IP constraint $\eta=0.01$ when $I$ is larger than $100$. This indicates that the IP performance with the proposed algorithm can be optimized by increasing $I$. With the algorithm in \cite{HetNet_CR_Guodong}, the IP always remains constant as $I$ grows. By combining Fig. \ref{AP_blocks} and Fig. \ref{IP_blocks}, we demonstrate that the increase of $I$ may improve the AP and reduce the IP with the proposed algorithm. However, the algorithm in \cite{HetNet_CR_Guodong} and the SLI algorithm are not sensitive to $I$. In other words, the C-SBS cannot satisfy the IP constraint by increasing $I$. This limits the application of the algorithm in \cite{HetNet_CR_Guodong} and the SLI algorithm.

            \begin{figure}[t!]
            \centering
            \includegraphics[scale=0.5]{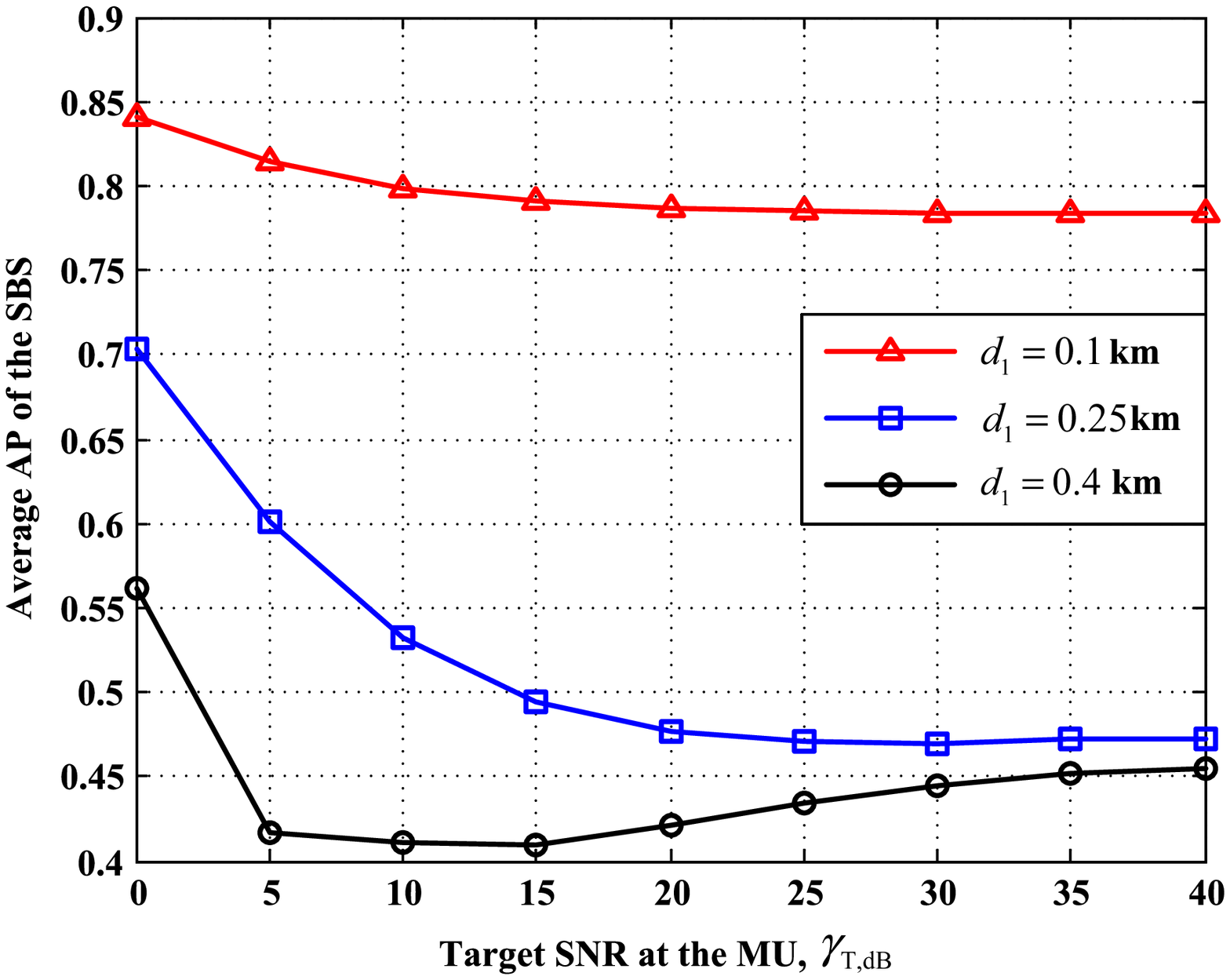}
            \caption{AP with different target SNRs at the MU in different scenarios.}
            \label{AP_gamma_T}
            \end{figure}

                    \begin{figure}[t!]
            \centering
            \includegraphics[scale=0.5]{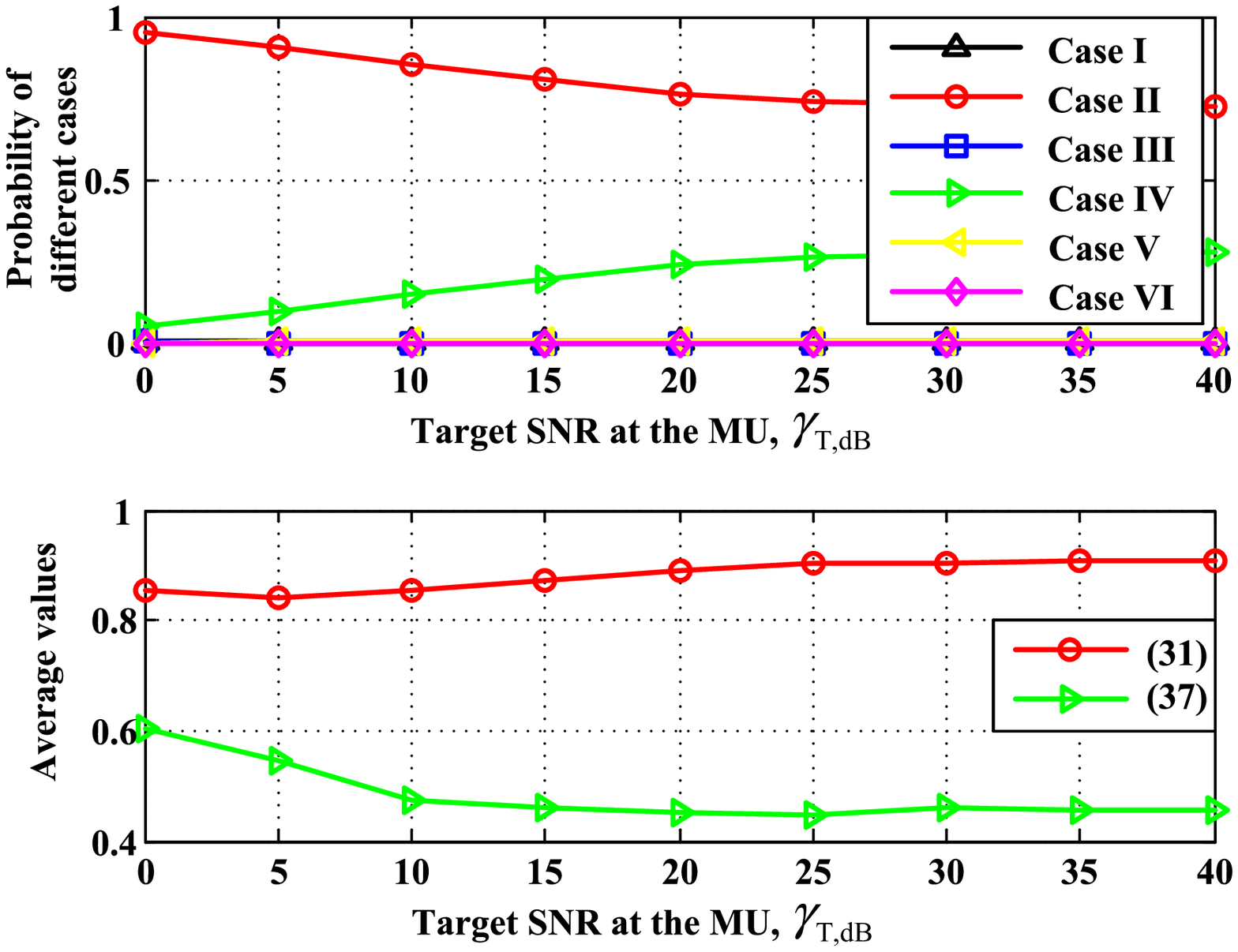}
            \caption{The upper subfigure provides the probabilities of different cases for $d_1=0.1$ km in the simulation; the lower subfigure provides the average values of (\ref{I_II_AP}) in Case II and (\ref{I_IV_AP}) in Case IV in the simulation.}
            \label{Prob_gamma_T}
            \end{figure}

Fig. \ref{AP_gamma_T} studies the AP with different target SNRs at MUs in different scenarios. In particular, when $d_1=0.1$ km, i.e., Scenario I, the AP of the C-SBS is calculated by (\ref{I_I_AP}) in Case I, (\ref{I_II_AP}) in Case II, (\ref{I_III_AP}) in Case III, (\ref{I_IV_AP}) in Case IV, (\ref{I_V_AP}) in Case V, or (\ref{I_VI_AP}) in Case VI. We observe that the AP of the C-SBS decreases as $\gamma_{T}$ grows. Since the AP is related to six cases, it is difficult to analyze the trend of the AP theoretically. Instead, we study the probabilities of different cases in Fig. \ref{Prob_gamma_T}.  From Fig. \ref{Prob_gamma_T}, the AP is mainly determined by Case II and Case IV. In particular, the probabilities of Case II and Case IV decrease and increase as $\gamma_{T}$ grows, respectively. We also provide the average values of (\ref{I_II_AP}) in Case II and (\ref{I_IV_AP}) in Case IV. Since the average value of (\ref{I_II_AP}) in Case II is much larger than that of (\ref{I_IV_AP}) in Case IV, the overall decreases as $\gamma_{T}$ grows. The AP curves for $d_1=0.25$ km (Scenario II) and $d_1=0.4$ km (Scenario III) can be similarly analyzed.

%
%

                \begin{figure}[t!]
            \centering
            \includegraphics[scale=0.5]{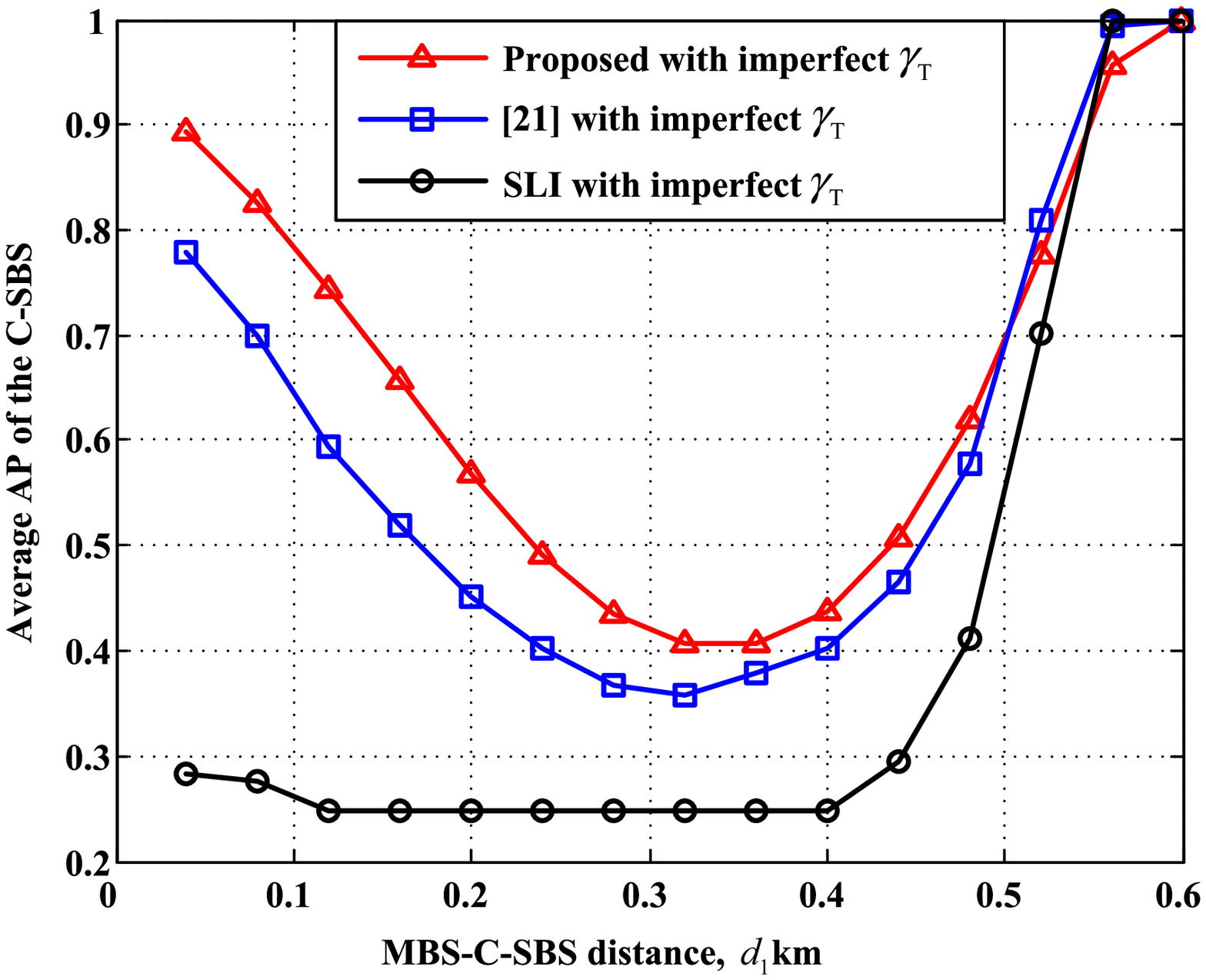}
            \caption{AP performance of different algorithms with imperfect target SNR $\gamma_T$. In particular, real $\gamma_{T,dB}$ is uniformly distributed between $17$ dB and $23$ dB, and the estimated target SNR $\gamma_{T,dB}=20$ dB.}
            \label{AP_d1_imperfect}
        \end{figure}

                \begin{figure}[t!]
            \centering
            \includegraphics[scale=0.5]{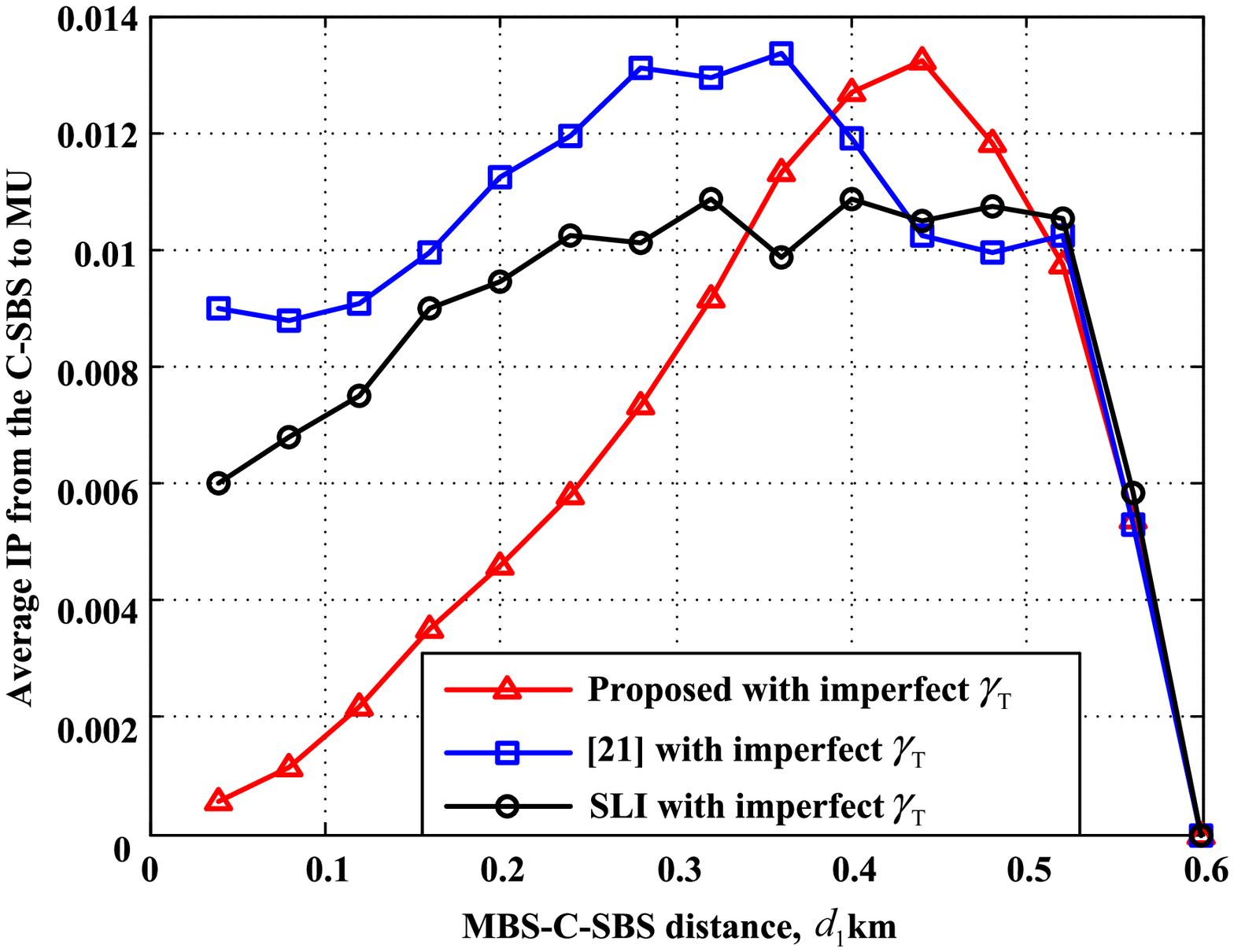}
            \caption{IP performance of different algorithms with imperfect target SNR $\gamma_T$. In particular, real $\gamma_{T,dB}$ is uniformly distributed between $17$ dB and $23$ dB, the estimated target SNR $\gamma_{T,dB}=20$ dB.}
            \label{IP_d1_imperfect}
        \end{figure}

In the proposed algorithm, the target SNR $\gamma_{T,dB}$ of the MU is needed at the C-SBS to design the AP and satisfy the IP constraint at the MU. However, the information of $\gamma_{T,dB}$ is unknown at the C-SBS. Although $\gamma_{T,dB}$ can be estimated at the C-SBS by observing the \emph{modulation and coding scheme} (MCS) of the MBS signal, there may exist some estimation errors of $\gamma_{T,dB}$. Here, we consider imperfect $\gamma_{T,dB}$ at the C-SBS and compare the performance of different algorithms. In particular, we assume that $\gamma_{T,dB}$ is uniformly distributed between $17$ dB and $23$ dB. In other words, the estimated target SNR $\gamma_{T,dB}=20$ dB has up to $\pm 3$ dB errors. The AP and IP performance is provided in Fig. \ref{AP_d1_imperfect} and Fig. \ref{IP_d1_imperfect}, respectively. By comparing the AP performance with perfect target SNR $\gamma_{T,dB}$ in Fig. \ref{AP_d1} and the AP performance with imperfect target SNR $\gamma_{T,dB}$ in Fig. \ref{AP_d1_imperfect}, imperfect target SNR $\gamma_{T,dB}$ has little impact on the AP performance. Meanwhile, by comparing the IP performance with perfect target SNR $\gamma_{T,dB}$ in Fig. \ref{IP_d1} and the IP performance with imperfect target SNR $\gamma_{T,dB}$ in Fig. \ref{IP_d1_imperfect}, imperfect target SNR $\gamma_{T,dB}$ may slightly increase the IP with the proposed algorithm and the algorithm in \cite{HetNet_CR_Guodong}, and has little impact on the SLI algorithm (Since the SLI algorithm is independent on target SNR $\gamma_{T,dB}$). In particular, the IP with the proposed algorithm is below the IP constraint $\eta=0.01$ except for $0.34 \leq d \leq 0.52$ km, and the IP with the algorithm in \cite{HetNet_CR_Guodong} is below the IP constraint $\eta=0.01$ except for $0.34 \leq d \leq 0.52$ km, and the IP with the SLI algorithm is below the IP constraint $\eta=0.01$ except for $0.24 \leq d \leq 0.52$ km. Thus, the proposed algorithm can be applied in a wider region compared with the algorithms in \cite{HetNet_CR_Guodong} and the SLI algorithm, even when the target SNR $\gamma_{T,dB}$ is imperfect.

Besides, to satisfy the IP constraint with the proposed algorithm for $0.34 \leq d \leq 0.52$ km in the practical deployment, two schemes may be adopted. The first one is to increase the number of blocks as shown in Fig. \ref{AP_blocks} (if possible). The second one is to reduce the AP for $0.34 \leq d \leq 0.52$ km. For instance, for $d_1=0.44$ km, the AP and IP are around $0.5$ and $0.013$, respectively. Then, we reduce the AP as $\rho_{\text{AP}}=0.5\times \frac{\eta}{0.013}\approx 0.4$ when considering imperfect target SNR $\gamma_{T,dB}$.

\section{Conclusions}
In this paper, we proposed a learning-based method for the C-SBS to exploit the MBS-MU distance information and realize the coexistence between a cognitive small cell and a macro cell in a two-tier underlay HetNet. In particular, we first enabled the C-SBS to analyze the MBS signal and learn the MBS-MU distance information. Then, we calculated the upper bound of the probability that the MU is in the C-SBS coverage and design an AP with a closed-form expression to satisfy the IP constraint at the MU. Numerical results indicated that the proposed algorithm outperforms the existing methods up to $60\%$ AP (or transmission opportunity) improvement. With the proposed algorithms, the cognitive small cell can use the same frequency bands as the macro cell at the same time in a HetNet. This enhances the spectrum efficiency and provides a potential solution for the spectrum scarcity problem in the future wireless communications. Besides, the cognitive small cell can coexist with the macro cell without any centralized coordinator. This reduces the cost to deploy a cognitive small cell within a HetNet and is meaningful from the practical perspective. Furthermore, we demonstrate that the MBS-MU distance information is of great importance for a cognitive small cell to effectively coexist with the macro cell.

\section{Appendix}

\subsection{Proof of Lemma 1}

To prove Lemma 1, we only need to verify $F_{\Gamma_{1,dB}}\left(\gamma_{T,dB}+37.6\log_{10}\left(\frac{d_0}{d_1}\right)\right)=\frac{1}{2}$. From (\ref{gamma_c_dB_CDF_2}), we have
\begin{align} \nonumber
&F_{\Gamma_{1,dB}}\left(\gamma_{T,dB}+37.6\log_{10}\left(\frac{d_0}{d_1}\right)\right)\\ \nonumber
=&\Pr\left\{\gamma_{1,dB} \leq \gamma_{T,dB}+37.6\log\left(\frac{d_0}{d_1}\right) \right\}\\ \nonumber
=&\int_{ - \infty }^\infty  f_{\Theta _s}\left(\theta _s \right) F_{\Theta _r}\left(- {\theta _s} \right)d\theta _s\\ \nonumber
= &\!\int_{ - \infty }^0  \!\!\!f_{\Theta _s}\!\! \left(\theta _s \right)\frac{1}{1 \!+\! 10^{\frac{\theta _s}{10}}}d\theta _s\!\!+\!\!\int_{0 }^\infty \!\!\! f_{\Theta _s}\!\!\left(\theta _s \right)\frac{1}{1\! + \! 10^{\frac{\theta _s}{10}}}d\theta _s\\
\!=&\!\int_{ 0}^\infty  \!\!\! f_{\Theta _s}\!\!\left(-\theta _s \!\right)\frac{1}{1\! + \!10^{-\frac{\theta _s}{10}}}d\theta _s\!\!+\!\!\int_{0 }^\infty \!\!\! f_{\Theta _s}\!\!\left(\theta _s \!\right)\frac{1}{1 \!+ \!10^{\frac{\theta _s}{10}}}d\theta _s.
\label{Proof_1}
\end{align}

From (\ref{Theta_2_pdf}), we observer that ${{f_{{\Theta _s}}}\left( {{\theta _s}} \right)}$ is an even function. Then, we have $f_{\Theta _s}\left(-\theta _s \right)=f_{\Theta _s}\left(\theta _s \right)$. Meanwhile, we have $\frac{1}{1 + 10^{-\frac{\theta _s}{10}}}=1-\frac{1}{1 + 10^{\frac{\theta _s}{10}}}$. Thus, (\ref{Proof_1}) can be rewritten as
\begin{align}\label{Proof_2}
\nonumber
&F_{\Gamma_{1,dB}}\left(\gamma_{T,dB}+37.6\log_{10}\left(\frac{d_0}{d_1}\right)\right)\\ \nonumber
=&\int_{ 0}^\infty \!\! \! f_{\Theta _s}\left(\theta _s \right)\left(\!1-\!\frac{1}{1 \!+ \!10^{\frac{\theta _s}{10}}}\!\right)d\theta _s\!+\!\int_{0 }^\infty  \!\!\!f_{\Theta _s}\left(\theta _s \right)\frac{1}{1 \!+ \! 10^{\frac{\theta _s}{10}}}d\theta _s\\ \nonumber
=&\int_{0 }^\infty  f_{\Theta _s}\left(\theta _s \right)d\theta _s\\
=&\frac{1}{2}.
\end{align}
Here, we complete the proof of Lemma 1.

\subsection{Proof of Theorem 1}
Based on the value of $d_0^e$, we discuss three cases as follows:

$\bullet$ \emph{For the case $f(\bar \gamma_{1, dB}(1)) >d_0^e >0$}: In this case, it is straightforward to obtain
\begin{align}
1=\Pr\left\{d_0 > 0  \right\}
>\Pr\left\{d_0 > d_0^e \right\}
> \Pr\left\{d_0 > f(\bar \gamma_{1, dB}(1)) \right\}.
\label{C_bound_1}
\end{align}

Then, we calculate $\Pr\left\{d_0 > f(\bar \gamma_{1, dB}(1))\right\}$ in the following.

From the proof of Lemma 1, the probability that $\gamma_{1, dB}(k)$ ($1\leq k \leq K$) is larger than or equal to $\gamma _{T,dB}+37.6{\log _{10}}\left( \frac{d_0}{d_1} \right)$ is $\frac{1}{2}$. Then, the probability that $\bar \gamma_{1, dB}(1)$ is larger than $\gamma _{T,dB}+37.6{\log _{10}}\left( \frac{d_0}{d_1} \right)$ is equal to the probability that $K$ measured SNRs in $\bar \gamma_{1, dB}(k)$ $\left(1 \leq k \leq K\right)$ are larger than $\gamma _{T,dB}+37.6{\log _{10}}\left( \frac{d_0}{d_1} \right)$, i.e.,
\begin{align}\nonumber
&\Pr\left\{\bar \gamma_{1, dB}(1) > \gamma _{T,dB}+37.6{\log _{10}}\left( \frac{d_0}{d_1} \right)\right\}\\ \nonumber
=& \left(\Pr\left\{\bar \gamma_{1, dB}(k)\!\!> \!\!\gamma _{T,dB}+37.6{\log _{10}}\left( \frac{d_0}{d_1} \right)\right\}\right)^{K}\\
=& \left(\frac{1}{2}\right)^{K},
\label{C_bound_1_1}
\end{align}
which can be rewritten as
\begin{align}
\Pr\left\{d_0 < f(\bar \gamma_{1, dB}(1))\right\}=\left(\frac{1}{2}\right)^{K}.
\label{C_bound_1_2}
\end{align}

Thus, we have
\begin{align}
\Pr\left\{d_0 \geq f(\bar \gamma_{1, dB}(1))\right\}=1- \left(\frac{1}{2}\right)^{K}.
\label{C_bound_1_3}
\end{align}

By combining (\ref{C_bound_1}) and (\ref{C_bound_1_2}), we have (\ref{Th2_bound1}).

$\bullet$ \emph{For the case $R>d_0^e\geq f(\bar \gamma_{1, dB}(K))$}: In this case, it is straightforward to obtain
\begin{align}
0= \Pr\left\{R \leq d_0\right\}
< \Pr\left\{d_0^e \leq d_0\right\}
\leq \Pr\left\{f(\bar \gamma_{1, dB}(K)) \leq d_0 \right\}.
\label{C_bound_2}
\end{align}

We note that the probability that $\gamma_{1, dB}(K)$ is no larger than $\gamma _{T,dB}+37.6{\log _{10}}\left( \frac{d_0}{d_1} \right)$ is equal to the probability that $K$ measured SNRs in $\bar \gamma_{1, dB}(k)$ $\left(1 \leq k \leq K\right)$ are no larger than $\gamma _{T,dB}+37.6{\log _{10}}\left( \frac{d_0}{d_1} \right)$, i.e.,
\begin{align} \nonumber
&\Pr\left\{\bar \gamma_{1, dB}(K) \leq \gamma _{T,dB}+37.6{\log _{10}}\left( \frac{d_0}{d_1} \right)\right\}\\ \nonumber
=&\left(\Pr\left\{\bar \gamma_{1, dB}(k)\!\!\leq \!\!\gamma _{T,dB}+37.6{\log _{10}}\left( \frac{d_0}{d_1} \right)\right\}\right)^{K} \\
=&\left(\frac{1}{2}\right)^{K},
\label{C_bound_2_1}
\end{align}
which can be rewritten as
\begin{align}
\Pr\left\{f(\bar \gamma_{1, dB}(K)) \leq d_0 \right\}=\left(\frac{1}{2}\right)^{K}.
\label{C_bound_2_3}
\end{align}

By combining (\ref{C_bound_2}) and (\ref{C_bound_2_3}), we have (\ref{Th2_bound2}).

$\bullet$ \emph{For $f(\bar \gamma _{1,dB}(k'))\leq d_0^e < f(\bar \gamma _{1,dB}(k'+1))$, where $1\leq k' \leq K-1$}:
In this case, it is straightforward to obtain
\begin{align} \nonumber
\Pr\left\{f(\bar \gamma _{1,dB}(k'+1)) \leq d_0 \right\} &< \Pr\left\{d_0^e \leq d_0 \right\}\\
&\leq \Pr\left\{f(\bar \gamma _{1,dB}(k')) \leq d_0 \right\}.
\label{C_bound_3}
\end{align}

We note that the probability that $\bar \gamma_{1, dB}(k')$ is no larger than $\gamma _{T,dB}+37.6{\log _{10}}\left( \frac{d_0}{d_1} \right)$ is equal to the probability that at least $k'$ out of $K$ measured SNRs in $\gamma_{1, dB}(k)$ $\left(1 \leq k \leq K\right)$ are no larger than $\gamma _{T,dB}+37.6{\log _{10}}\left( \frac{d_0}{d_1} \right)$, i.e.,
\begin{align} \nonumber
&\Pr\left\{\bar \gamma_{1, dB}(k') \leq \gamma _{T,dB}+37.6{\log _{10}}\left(\frac{d_0}{d_1} \right)\right\}\\ \nonumber
=&C_K^{k'}\left(\frac{1}{2}\right)^{k'}\left(\frac{1}{2}\right)^{K-k'}+ C_K^{k'+1}\left(\frac{1}{2}\right)^{k'+1}\left(\frac{1}{2}\right)^{K-k'-1} \\ \nonumber
&+\cdots+C_K^{K}\left(\frac{1}{2}\right)^{K}\\
=&\left(\frac{1}{2}\right)^K\left(C_K^{k'}+C_K^{k'+1}+\cdots+C_K^{K}\right).
\label{C_bound_3_1}
\end{align}

Then, we have
\begin{align} \nonumber
&\Pr\left\{d_110^{\frac{\bar \gamma _{1,dB(k')-\gamma _{T,dB}}}{37.6}} \leq d_0 \right\} \\ =& \left(\frac{1}{2}\right)^K\left(C_K^{k'}+C_K^{k'+1}+\cdots+C_K^{K}\right).
\label{C_bound_3_2}
\end{align}

Similarly, we have
\begin{align}\nonumber
&\Pr\left\{d_110^{\frac{\bar \gamma _{1,dB(k'+1)-\gamma _{T,dB}}}{37.6}} \leq d_0 \right\}\\
=& \left(\frac{1}{2}\right)^K\left(C_K^{k'+1}+C_K^{k'+2}+\cdots+C_K^{K}\right).
\label{C_bound_3_3}
\end{align}

Combining (\ref{C_bound_3}), (\ref{C_bound_3_2}), and (\ref{C_bound_3_3}), we have (\ref{Th2_bound3}). Here, we complete the proof of this theorem.

\subsection{Derivation of $S_{\text{c}}$ in Scenario I }

            \begin{figure}[t!]
            \centering
            \includegraphics[scale=0.7]{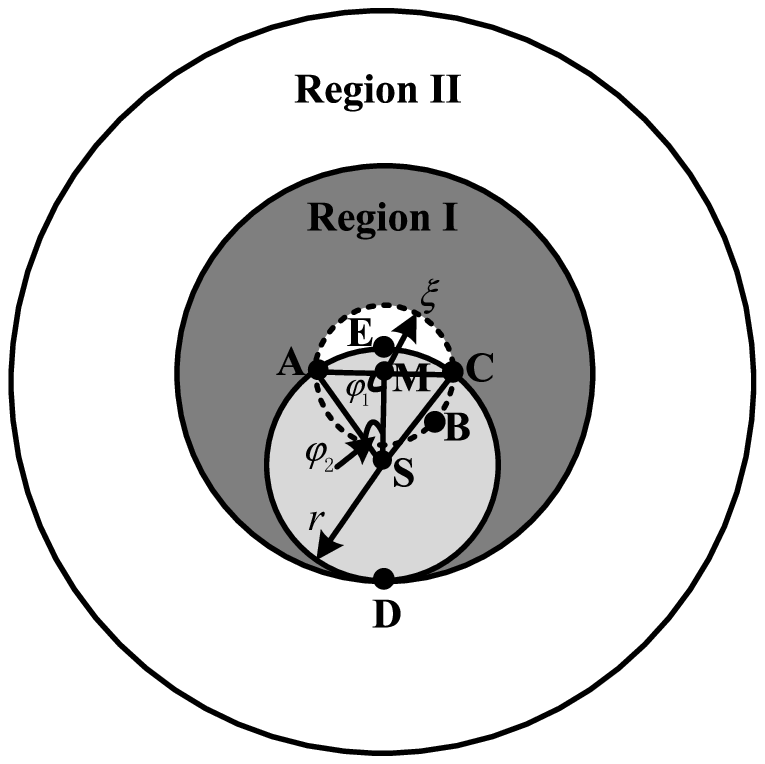}
            \caption{Geometrical model to calculate $S_{\text{c}}$ in Scenario I, i.e., $\xi \leq d_1 \leq r+\xi$. }
            \label{p_c_center}
        \end{figure}

In this part, we consider that the MU is in Scenario I and calculate the area $S_{\text{c}}$. Without loss of generality, we show the geometrical model in Fig. \ref{p_c_center}. In the sequential, we denote $\text{o}(\iota)$ as the circle region centered by point ``o'' with radius $\iota$.

When $\xi \leq d_1 \leq r-\xi$, $\text{M}(\xi)$ is in the $\text{S}(r)$. Then, we have
\begin{equation}
S_{\text{c}}=\pi r^2-\pi \xi^2.
\label{S_c_SI_1}
\end{equation}

When $r-\xi < d_1 \leq r+\xi$, $\text{M}(\xi)$ intersects with $\text{S}(r)$. We first draw auxiliary segments MA, MC, SA, SC, and MS. Then, we have MA=MC=$\xi$, SA=SC=$r$, and MS=$d_1$. Then, we denote ``B'' as a point both on the edge of $\text{M}(\xi)$ and in $\text{S}(r)$, denote ``D'' as a point on the edge of both $\text{S}(r)$ and $\text{M}(d_1+r)$, and denote ``E'' as a point both on the edge of $\text{S}(r)$ and in $\text{M}(\xi)$. Besides, we denote $\angle \text{AMS}=\varphi_1$ and $\angle \text{ASM}=\varphi_2$. According to the cosine theorem, we have
\begin{equation}
r^2=\xi^2+d_1^2-2\xi d_1\cos \varphi_1
\end{equation}
and
\begin{equation}
\xi^2=r^2+d_1^2-2rd_1\cos \varphi_2.
\end{equation}

Thus, we obtain $\varphi_1$ and $\varphi_2$ as
\begin{equation}
\varphi_1=\arccos \frac{\xi^2+d_1^2-r^2}{2\xi d_1}
\end{equation}
and
\begin{align}
\varphi_2= \arccos \frac{r^2+d_1^2-\xi^2}{2rd_1}.
\end{align}
Then, we have
 \begin{align} \nonumber
S_{\text{c}}=&S_{\text{ABCD}}\\ \nonumber
=& \pi r^2- S_{\text{ABCE}}\\ \nonumber
=& \pi r^2 -(S_{\text{ABCM}}+S_{\text{ASCE}}-S_{\text{ASCM}})\\ \nonumber
= & \pi r^2- \varphi_1 \xi^2 -\varphi_2 r^2+\xi d_1 \sin \varphi_1\\
= & (\pi-\varphi_2) r^2- \varphi_1 \xi^2 +\xi d_1 \sin \varphi_1.
\label{S_c_SI_2}
\end{align}

By combining (\ref{S_c_SI_1}) and (\ref{S_c_SI_2}), we obtain (\ref{S_c_SI}).

\subsection{AP design in Scenario II}
In this part, we analyze each case in Scenario II separately. In particular, we first calculate the probability $\rho_{\text{II}}$ that the MU is in Region II for each case and then design the corresponding AP. In the sequential, we denote $S_{\text{II}}=\pi (d_1+r)^2- \pi (d_1-r)^2$ as the area of Region II and denote $S_{\text{c}}=\pi r^2$ as the area of the interference region, where the MU may appear in the C-SBS coverage.

\subsubsection{AP Design in Case I}
In this case, $d_1-r$  satisfies (\ref{II_k_1}) and $d_1+r$  satisfies (\ref{I_k_2}). The probability $\rho_{\text{II}}$ that the MU is in Region II is
\begin{align} \nonumber
\rho_{\text{II}}=&\Pr\left\{d_1-r\leq d_0 \leq d_1+r\right\}\\ \nonumber
=& \Pr\left\{d_0> d_1-r\right\}-\Pr\left\{d_0> d_1+r\right\} \\ \nonumber
\overset{\text{(II-a)}}{\leq} & 1- 1+\left(\frac{1}{2}\right)^K\\
=& \left(\frac{1}{2}\right)^K,
\label{II_I_bound1}
\end{align}
where we use (\ref{Th2_bound1}) in (II-a).

To control the IP constraint $\eta$ at the MU, the AP of the C-SBS needs to satisfy $0\leq \rho_{\text{AP}}\leq 1$ and $\rho_{\text{AP}}\rho_{\text{II}}\frac{S_{\text{c}}}{S_{\text{II}}} \leq \eta$, i.e., $\rho_{\text{AP}} \leq \min \left\{\frac{\eta S_{\text{II}}}{\rho_{\text{II}}S_{\text{c}}},1\right\}=\min \left\{\frac{\eta \left[(d_1+r)^2-(d_1-r)^2\right]}{\rho_{\text{II}}r^2},1\right\}$, which can be lower bounder by
\begin{align} \nonumber
&\min \left\{\frac{\eta \left[(d_1+r)^2-(d_1-r)^2\right]}{\rho_{\text{II}}r^2},1\right\} \\
\overset{\text{(II-b)}}{\geq} & \min \left\{\frac{\eta2^K \left[(d_1+r)^2-(d_1-r)^2\right]}{r^2},1\right\},
\label{II_I_bound2}
\end{align}
where we use (\ref{II_I_bound1}) in (II-b).

Thus, to protect the MU in this case, the maximum AP of the C-SBS is
\begin{align}
 \rho_{\text{AP}}=\min \left\{\frac{\eta2^K \left[(d_1+r)^2-(d_1-r)^2\right]}{r^2},1\right\}.
\label{II_I_AP}
\end{align}

\subsubsection{AP Design in Case II}
In this case, $d_1-r$  satisfies (\ref{II_k_1}) and $d_1+r$  satisfies (\ref{I_k_2_1}). The probability $\rho_{\text{II}}$ that the MU is in Region II is
\begin{align} \nonumber
\rho_{\text{II}}=& \Pr\left\{d_0> d_1-r\right\}-\Pr\left\{d_0> d_1+r\right\} \\ \nonumber
\overset{\text{(II-c)}}{\leq} & 1- \left(\frac{1}{2}\right)^K\left(C_K^{k_2+1}+C_K^{k_2+2}+\cdots+C_K^{K}\right)\\
=& \left(\frac{1}{2}\right)^K\left(C_K^{0}+C_K^{1}+\cdots+C_K^{k_2}\right),
\label{II_II_bound1}
\end{align}
where we use (\ref{Th2_bound1}) and (\ref{Th2_bound3}) in (II-c).

To control the IP constraint $\eta$ at the MU, the AP of the C-SBS needs to satisfy $0\leq \rho_{\text{AP}}\leq 1$ and $\rho_{\text{AP}}\rho_{\text{II}}\frac{S_{\text{c}}}{S_{\text{II}}} \leq \eta$, i.e., $\rho_{\text{AP}} \leq \min \left\{\frac{\eta S_{\text{II}}}{\rho_{\text{II}}S_{\text{c}}},1\right\}=\min \left\{\frac{\eta \left[(d_1+r)^2-(d_1-r)^2\right]}{\rho_{\text{II}}r^2},1\right\}$, which can be lower bounder by
\begin{align} \nonumber
 & \min \left\{\frac{\eta \left[(d_1+r)^2-(d_1-r)^2\right]}{\rho_{\text{II}}r^2},1\right\}\\
\overset{\text{(II-d)}}{\geq} &  \left\{\frac{\eta 2^K \left[(d_1+r)^2-(d_1-r)^2\right]}{\left(C_K^{0}+C_K^{1}+\cdots+C_K^{k_2}\right)r^2},1\right\},
\label{II_II_bound2}
\end{align}
where we use (\ref{II_II_bound1}) in (II-d).

Thus, to protect the MU in this case, the maximum AP of the C-SBS is
\begin{align}
\rho_{\text{AP}}= \left\{\frac{\eta 2^K \left[(d_1+r)^2-(d_1-r)^2\right]}{\left(C_K^{0}+C_K^{1}+\cdots+C_K^{k_2}\right)r^2},1\right\}.
\label{II_II_AP}
\end{align}

\subsubsection{AP Design in Case III}
In this case, $d_1-r$  satisfies (\ref{II_k_1}) and $d_1+r$  satisfies (\ref{I_k_2_2}). The probability $\rho_{\text{II}}$ that the MU is in Region II is
\begin{align}
\rho_{\text{II}}=& \Pr\left\{d_0> d_1-r\right\}-\Pr\left\{d_0> d_1+r\right\}
\overset{\text{(II-e)}}{\leq} 1,
\label{II_III_bound1}
\end{align}
where we use (\ref{Th2_bound1}) and (\ref{Th2_bound2}) in (II-e).

To control the IP constraint $\eta$ at the MU, the AP of the C-SBS needs to satisfy $0\leq \rho_{\text{AP}}\leq 1$ and $\rho_{\text{AP}}\rho_{\text{II}}\frac{S_{\text{c}}}{S_{\text{II}}} \leq \eta$, i.e., $\rho_{\text{AP}} \leq \min \left\{\frac{\eta S_{\text{II}}}{\rho_{\text{II}}S_{\text{c}}},1\right\}=\min \left\{\frac{\eta \left[(d_1+r)^2-(d_1-r)^2\right]}{\rho_{\text{II}}r^2},1\right\}$, which can be lower bounder by
\begin{align} \nonumber
&\min \left\{\frac{\eta \left[(d_1+r)^2-(d_1-r)^2\right]}{\rho_{\text{II}}r^2},1\right\} \\
\overset{\text{(II-f)}}{\geq} & \min \left\{\frac{\eta \left[(d_1+r)^2-(d_1-r)^2\right]}{r^2},1\right\},
\label{II_III_bound2}
\end{align}
where we use (\ref{II_III_bound1}) in (II-f).

Thus, to protect the MU in this case, the maximum AP of the C-SBS is
\begin{align}
\rho_{\text{AP}}= \min \left\{\frac{\eta \left[(d_1+r)^2-(d_1-r)^2\right]}{r^2},1\right\}.
\label{II_III_AP}
\end{align}

\subsubsection{AP Design in Case IV}
In this case, $d_1-r$  satisfies (\ref{II_k_1_1}) and $d_1+r$  satisfies (\ref{I_k_2_1}). The probability $\rho_{\text{II}}$ that the MU is in Region II is
\begin{align} \nonumber
\rho_{\text{II}}=& \Pr\left\{d_0> d_1-r\right\}-\Pr\left\{d_0> d_1+r\right\} \\  \nonumber
\overset{\text{(II-g)}}{\leq} & \left(\frac{1}{2}\right)^K\left(C_K^{k_1}+C_K^{k_1+1}+\cdots+C_K^{K}\right)\\ \nonumber &-\left(\frac{1}{2}\right)^K\left(C_K^{k_2+1}+C_K^{k_2+2}+\cdots+C_K^{K}\right)\\
=&\left(\frac{1}{2}\right)^K\left(C_K^{k_1}+C_K^{k_1+1}+\cdots+C_K^{k_2}\right),
\label{II_IV_bound1}
\end{align}
where we use (\ref{Th2_bound3}) in (II-g).

To control the IP constraint $\eta$ at the MU, the AP of the C-SBS needs to satisfy $0\leq \rho_{\text{AP}}\leq 1$ and $\rho_{\text{AP}}\rho_{\text{II}}\frac{S_{\text{c}}}{S_{\text{II}}} \leq \eta$, i.e., $\rho_{\text{AP}} \leq \min \left\{\frac{\eta S_{\text{II}}}{\rho_{\text{II}}S_{\text{c}}},1\right\}=\min \left\{\frac{\eta \left[(d_1+r)^2-(d_1-r)^2\right]}{\rho_{\text{II}}r^2},1\right\}$, which can be lower bounder by
\begin{align} \nonumber
& \min \left\{\frac{\eta \left[(d_1+r)^2-(d_1-r)^2\right]}{\rho_{\text{II}}r^2},1\right\} \\
\overset{(\text{II-h})}{\geq} & \min \left\{\frac{\eta 2^K \left[(d_1+r)^2-(d_1-r)^2\right]}{r^2\left(C_K^{k_1}+C_K^{k_1+1}+\cdots+C_K^{k_2}\right)},1\right\},
\label{II_IV_bound2}
\end{align}
where we use (\ref{II_IV_bound1}) in (II-h).

Thus, to protect the MU in this case, the maximum AP of the C-SBS is
\begin{align}
\rho_{\text{AP}}= \min \left\{\frac{\eta 2^K \left[(d_1+r)^2-(d_1-r)^2\right]}{r^2\left(C_K^{k_1}+C_K^{k_1+1}+\cdots+C_K^{k_2}\right)},1\right\}.
\label{II_IV_AP}
\end{align}

\subsubsection{AP Design in Case V}
In this case, $d_1-r$  satisfies (\ref{II_k_1_1}) and $d_1+r$  satisfies (\ref{I_k_2_2}). The probability $\rho_{\text{II}}$ that the MU is in Region II is
\begin{align} \nonumber
\rho_{\text{II}}=& \Pr\left\{d_0> d_1-r\right\}-\Pr\left\{d_0> d_1+r\right\} \\
\overset{\text{(II-i)}}{\leq} & \left(\frac{1}{2}\right)^K\left(C_K^{k_1}+C_K^{k_1+1}+\cdots+C_K^{K}\right),
\label{II_V_bound1}
\end{align}
where we use (\ref{Th2_bound2}) and (\ref{Th2_bound3}) in (II-i).

To control the IP constraint $\eta$ at the MU, the AP of the C-SBS needs to satisfy $0\leq \rho_{\text{AP}}\leq 1$ and $\rho_{\text{AP}}\rho_{\text{II}}\frac{S_{\text{c}}}{S_{\text{II}}} \leq \eta$, i.e., $\rho_{\text{AP}} \leq \min \left\{\frac{\eta S_{\text{II}}}{\rho_{\text{II}}S_{\text{c}}},1\right\}=\min \left\{\frac{\eta \left[(d_1+r)^2-(d_1-r)^2\right]}{\rho_{\text{II}}r^2},1\right\}$, which can be lower bounder by
\begin{align} \nonumber
&\min \left\{\frac{\eta \left[(d_1+r)^2-(d_1-r)^2\right]}{\rho_{\text{II}}r^2},1\right\} \\
\overset{(\text{II-j})}{\geq} & \min \left\{\frac{\eta 2^K \left[(d_1+r)^2-(d_1-r)^2\right]}{r^2\left(C_K^{k_1}+C_K^{k_1+1}+\cdots+C_K^{K}\right)},1\right\},
\label{II_IV_bound2}
\end{align}
where we use (\ref{II_V_bound1}) in (II-j).

Thus, to protect the MU in this case, the maximum AP of the C-SBS is
\begin{align}
\rho_{\text{AP}}= \min \left\{\frac{\eta 2^K \left[(d_1+r)^2-(d_1-r)^2\right]}{r^2\left(C_K^{k_1}+C_K^{k_1+1}+\cdots+C_K^{K}\right)},1\right\}.
\label{II_IV_AP}
\end{align}

\subsubsection{AP Design in Case VI}
In this case, $d_1-r$  satisfies (\ref{II_k_1_2}) and $d_1+r$  satisfies (\ref{I_k_2_2}). The probability $\rho_{\text{II}}$ that the MU is in Region II is
\begin{align}
\rho_{\text{II}} \!=\! \Pr\left\{d_0\!>\! d_1-r\right\}\!-\!\Pr\left\{d_0\!>\! d_1+r\right\}
\overset{\text{(II-k)}}{\leq}  \left(\frac{1}{2}\right)^K,
\label{II_VI_bound1}
\end{align}
where we use (\ref{Th2_bound2}) in (II-k).

To control the IP constraint $\eta$ at the MU, the AP of the C-SBS needs to satisfy $0\leq \rho_{\text{AP}}\leq 1$ and $\rho_{\text{AP}}\rho_{\text{II}}\frac{S_{\text{c}}}{S_{\text{II}}} \leq \eta$, i.e., $\rho_{\text{AP}} \leq \min \left\{\frac{\eta S_{\text{II}}}{\rho_{\text{II}}S_{\text{c}}},1\right\}=\min \left\{\frac{\eta \left[(d_1+r)^2-(d_1-r)^2\right]}{\rho_{\text{II}}r^2},1\right\}$, which can be lower bounder by
\begin{align} \nonumber
&\min \left\{\frac{\eta \left[(d_1+r)^2-(d_1-r)^2\right]}{\rho_{\text{II}}r^2},1\right\} \\
\overset{\text{(II-l)}}{\geq} & \min \left\{\frac{\eta 2^K \left[(d_1+r)^2-(d_1-r)^2\right]}{r^2},1\right\},
\label{II_VI_bound2}
\end{align}
where we use (\ref{II_VI_bound1}) in (II-l).

Thus, to protect the MU in this case, the maximum AP of the C-SBS is
\begin{align}
\rho_{\text{AP}}=\min \left\{\frac{\eta 2^K \left[(d_1+r)^2-(d_1-r)^2\right]}{r^2},1\right\}.
\label{II_VI_AP}
\end{align}

\subsection{AP Design in Scenario III}
In this part, we analyze each case in Scenario III separately. In particular, we first calculate the probability $\rho_{\text{I}}$ that the MU is in Region II for each case and then design the corresponding AP. In the sequential, we denote $S_{\text{II}}=\pi R^2- \pi (d_1-r)^2$ as the area of Region II and denote $S_{\text{c}}$ as the area of the interference region, where the MU may appear in the coverage of the C-SBS. From Appendix F, we obtain
\begin{align}
S_{\text{c}}=\varphi_3 R^2+\varphi_4 r^2-Rd_1\sin \varphi_3,
\label{S_c_SIII}
\end{align}
where $\varphi_3=\arccos \frac{R^2+d_1^2-r^2}{2Rd_1}$ and $\varphi_4= \arccos \frac{r^2+d_1^2-R^2}{2rd_1}$.

\subsubsection{AP Design in Case I}
In this case, $d_1-r$ satisfies (\ref{II_k_1}) and $R$ satisfies (\ref{III_k_3}). The probability $\rho_{\text{II}}$ that the MU is in Region II is
\begin{align} \nonumber
\rho_{\text{II}}=&\Pr\left\{d_1-r< d_0 \leq R\right\}\\ \nonumber
=& \Pr\left\{d_0> d_1-r\right\}-\Pr\left\{d_0> R\right\} \\ \nonumber
\overset{\text{(III-a)}}{\leq} & 1- 1+\left(\frac{1}{2}\right)^K\\
=& \left(\frac{1}{2}\right)^K,
\label{III_I_bound1}
\end{align}
where we use (\ref{Th2_bound1}) in (III-a).

To control the IP constraint $\eta$ at the MU, the AP of the C-SBS needs to satisfy $0\leq \rho_{\text{AP}}\leq 1$ and $\rho_{\text{AP}}\rho_{\text{II}}\frac{S_{\text{c}}}{S_{\text{II}}} \leq \eta$, i.e., $\rho_{\text{AP}} \leq \min \left\{\frac{\eta S_{\text{II}}}{\rho_{\text{II}}S_{\text{c}}},1\right\}=\min \left\{\frac{\eta \left[\pi R^2-\pi(d_1-r)^2\right]}{\rho_{\text{II}}S_{\text{c}}},1\right\}$, which can be lower bounder by
\begin{align} \nonumber
&\min \left\{\frac{\eta \left[\pi R^2-\pi(d_1-r)^2\right]}{\rho_{\text{II}}S_{\text{c}}},1\right\} \\
\overset{\text{(III-b)}}{\geq} & \min \left\{\frac{\eta2^K \left[\pi R^2-\pi (d_1-r)^2\right]}{S_{\text{c}}},1\right\},
\label{III_I_bound2}
\end{align}
where we use (\ref{III_I_bound1}) in (III-b).

Thus, to protect the MU in this case, the maximum AP of the C-SBS is
\begin{align}
 \rho_{\text{AP}}=\min \left\{\frac{\eta2^K \left[\pi R^2-\pi(d_1-r)^2\right]}{S_{\text{c}}},1\right\},
\label{III_I_AP}
\end{align}
where $S_{\text{c}}$ is shown in (\ref{S_c_SIII}).

\subsubsection{AP Design in Case II}
In this case, $d_1-r$  satisfies (\ref{II_k_1}) and $R$ satisfies (\ref{III_k_3_1}). The probability $\rho_{\text{II}}$ that the MU is in Region II is
\begin{align} \nonumber
\rho_{\text{II}}=& \Pr\left\{d_0> d_1-r\right\}-\Pr\left\{d_0> R\right\} \\ \nonumber
\overset{\text{(III-c)}}{\leq} & 1- \left(\frac{1}{2}\right)^K\left(C_K^{k_3+1}+C_K^{k_3+2}+\cdots+C_K^{K}\right)\\
=& \left(\frac{1}{2}\right)^K\left(C_K^{0}+C_K^{1}+\cdots+C_K^{k_3}\right),
\label{III_II_bound1}
\end{align}
where we use (\ref{Th2_bound1}) and (\ref{Th2_bound3}) in (III-c).

To control the IP constraint $\eta$ at the MU, the AP of the C-SBS needs to satisfy $0\leq \rho_{\text{AP}}\leq 1$ and $\rho_{\text{AP}}\rho_{\text{II}}\frac{S_{\text{c}}}{S_{\text{II}}} \leq \eta$, i.e., $\rho_{\text{AP}} \leq \min \left\{\frac{\eta S_{\text{II}}}{\rho_{\text{II}}S_{\text{c}}},1\right\}=\min \left\{\frac{\eta \left[\pi R^2-\pi(d_1-r)^2\right]}{\rho_{\text{II}}S_{\text{c}}},1\right\}$, which can be lower bounder by
\begin{align} \nonumber
&\min \left\{\frac{\eta \left[\pi R^2-\pi(d_1-r)^2\right]}{\rho_{\text{II}}S_{\text{c}}},1\right\}\\
\overset{\text{(III-d)}}{\geq} & \left\{\frac{\eta 2^K \left[\pi R^2-\pi (d_1-r)^2\right]}{\left(C_K^{0}+C_K^{1}+\cdots+C_K^{k_3}\right)S_{\text{c}}},1\right\},
\label{III_II_bound2}
\end{align}
where we use (\ref{III_II_bound1}) in (III-d).

Thus, to protect the MU in this case, the maximum AP of the C-SBS is
\begin{align}
\rho_{\text{AP}}= \left\{\frac{\eta 2^K \left[\pi R^2-\pi(d_1-r)^2\right]}{\left(C_K^{0}+C_K^{1}+\cdots+C_K^{k_3}\right)S_{\text{c}}},1\right\},
\label{II_II_AP}
\end{align}
where $S_{\text{c}}$ is shown in (\ref{S_c_SIII}).

\subsubsection{AP Design in Case III}
In this case, $d_1-r$  satisfies (\ref{II_k_1}) and $R$ satisfies (\ref{III_k_3_2}). The probability $\rho_{\text{II}}$ that the MU is in Region II is
\begin{align}
\rho_{\text{II}}=& \Pr\left\{d_0> d_1-r\right\}-\Pr\left\{d_0> R\right\}
\overset{\text{(III-e)}}{\leq}  1,
\label{III_III_bound1}
\end{align}
where we use (\ref{Th2_bound1}) and (\ref{Th2_bound2}) in (III-e).

To control the IP constraint $\eta$ at the MU, the AP of the C-SBS needs to satisfy $0\leq \rho_{\text{AP}}\leq 1$ and $\rho_{\text{AP}}\rho_{\text{II}}\frac{S_{\text{c}}}{S_{\text{II}}} \leq \eta$, i.e., $\rho_{\text{AP}} \leq \min \left\{\frac{\eta S_{\text{II}}}{\rho_{\text{II}}S_{\text{c}}},1\right\}=\min \left\{\frac{\eta \left[\pi R^2-\pi(d_1-r)^2\right]}{\rho_{\text{II}}S_{\text{c}}},1\right\}$, which can be lower bounder by
\begin{align}\nonumber
& \min \left\{\frac{\eta \left[\pi R^2-\pi(d_1-r)^2\right]}{\rho_{\text{II}}S_{\text{c}}},1\right\}\\
 \overset{\text{(III-f)}}{\geq} &\min \left\{\frac{\eta \left[\pi R^2-\pi(d_1-r)^2\right]}{S_{\text{c}}},1\right\},
\label{III_III_bound2}
\end{align}
where we use (\ref{I_III_bound1}) in (III-f).

Thus, to protect the MU in this case, the maximum AP of the C-SBS is
\begin{align}
\rho_{\text{AP}}=\min \left\{\frac{\eta \left[\pi R^2-\pi(d_1-r)^2\right]}{S_{\text{c}}},1\right\},
\label{III_III_AP}
\end{align}
where $S_{\text{c}}$ is shown in (\ref{S_c_SIII}).

\subsubsection{AP Design in Case IV}
In this case, $d_1-r$  satisfies (\ref{II_k_1_1}) and $R$ satisfies (\ref{III_k_3_1}). The probability $\rho_{\text{II}}$ that the MU is in Region II is
\begin{align} \nonumber
\rho_{\text{II}}=& \Pr\left\{d_0> d_1-r\right\}-\Pr\left\{d_0> R\right\} \\  \nonumber
\overset{\text{(III-g)}}{\leq} & \left(\frac{1}{2}\right)^K\left(C_K^{k_1}+C_K^{k_1+1}+\cdots+C_K^{K}\right)\\ \nonumber
& -\left(\frac{1}{2}\right)^K\left(C_K^{k_3+1}+C_K^{k_3+2}+\cdots+C_K^{K}\right)\\
=&\left(\frac{1}{2}\right)^K\left(C_K^{k_1}+C_K^{k_1+1}+\cdots+C_K^{k_3}\right),
\label{III_IV_bound1}
\end{align}
where we use (\ref{Th2_bound3}) in (III-g).

To control the IP constraint $\eta$ at the MU, the AP of the C-SBS needs to satisfy $0\leq \rho_{\text{AP}}\leq 1$ and $\rho_{\text{AP}}\rho_{\text{II}}\frac{S_{\text{c}}}{S_{\text{II}}} \leq \eta$, i.e., $\rho_{\text{AP}} \leq \min \left\{\frac{\eta S_{\text{II}}}{\rho_{\text{II}}S_{\text{c}}},1\right\}=\min \left\{\frac{\eta \left[\pi R^2-\pi(d_1-r)^2\right]}{\rho_{\text{II}}S_{\text{c}}},1\right\}$, which can be lower bounder by
\begin{align} \nonumber
& \min \left\{\frac{\eta \left[\pi R^2-\pi(d_1-r)^2\right]}{\rho_{\text{II}}S_{\text{c}}},1\right\}\\
\overset{\text{(III-h)}}{\geq} & \min \left\{\frac{\eta 2^K \left[\pi R^2-\pi(d_1-r)^2\right]}{\left(C_K^{k_1}+C_K^{k_1+1}+\cdots+C_K^{k_3}\right)S_{\text{c}}},1\right\},
\label{III_IV_bound2}
\end{align}
where we use (\ref{III_IV_bound1}) in (III-h).

Thus, to protect the MU in this case, the maximum AP of the C-SBS is
\begin{align}
\rho_{\text{AP}}= \min \left\{\frac{\eta 2^K \left[\pi R^2-\pi(d_1-r)^2\right]}{\left(C_K^{k_1}+C_K^{k_1+1}+\cdots+C_K^{k_3}\right)S_{\text{c}}},1\right\},
\label{III_IV_AP}
\end{align}
where $S_{\text{c}}$ is shown in (\ref{S_c_SIII}).

\subsubsection{AP Design in Case V}
In this case, $d_1-r$  satisfies (\ref{II_k_1_1}) and $R$ satisfies (\ref{III_k_3_2}). The probability $\rho_{\text{II}}$ that the MU is in Region II is
\begin{align} \nonumber
\rho_{\text{II}}=& \Pr\left\{d_0> d_1-r\right\}-\Pr\left\{d_0> R\right\} \\
\overset{\text{(III-i)}}{\leq} & \left(\frac{1}{2}\right)^K\left(C_K^{k_1}+C_K^{k_1+1}+\cdots+C_K^{K}\right),
\label{III_V_bound1}
\end{align}
where we use (\ref{Th2_bound2}) and (\ref{Th2_bound3}) in (III-i).

To control the IP constraint $\eta$ at the MU, the AP of the C-SBS needs to satisfy $0\leq \rho_{\text{AP}}\leq 1$ and $\rho_{\text{AP}}\rho_{\text{II}}\frac{S_{\text{c}}}{S_{\text{II}}} \leq \eta$, i.e., $\rho_{\text{AP}} \leq \min \left\{\frac{\eta S_{\text{II}}}{\rho_{\text{II}}S_{\text{c}}},1\right\}=\min \left\{\frac{\eta \left[\pi R^2-\pi(d_1-r)^2\right]}{\rho_{\text{II}}S_{\text{c}}},1\right\}$, which can be lower bounder by
\begin{align} \nonumber
& \min \left\{\frac{\eta \left[\pi R^2-\pi(d_1-r)^2\right]}{\rho_{\text{II}}S_{\text{c}}},1\right\}\\
\overset{\text{(III-j)}}{\geq} & \min \left\{\frac{\eta 2^K \left[\pi R^2-\pi(d_1-r)^2\right]}{\left(C_K^{k_1}+C_K^{k_1+1}+\cdots+C_K^{K}\right)S_{\text{c}}},1\right\},
\label{III_V_bound2}
\end{align}
where we use (\ref{III_V_bound1}) in (III-j).

Thus, to protect the MU in this case, the maximum AP of the C-SBS is
\begin{align}
\rho_{\text{AP}}= \min \left\{\frac{\eta 2^K \left[\pi R^2-\pi(d_1-r)^2\right]}{\left(C_K^{k_1}+C_K^{k_1+1}+\cdots+C_K^{K}\right)S_{\text{c}}},1\right\},
\label{III_V_AP}
\end{align}
where $S_{\text{c}}$ is shown in (\ref{S_c_SIII}).

\subsubsection{AP Design in Case VI}
In this case, $d_1-r$  satisfies (\ref{II_k_1_2}) and $R$ satisfies (\ref{III_k_3_2}). The probability $\rho_{\text{II}}$ that the MU is in Region II is
\begin{align}
\rho_{\text{II}}= \Pr\left\{d_0> d_1-r\right\}-\Pr\left\{d_0> R\right\}
\overset{\text{(III-k)}}{\leq}  \left(\frac{1}{2}\right)^K,
\label{III_VI_bound1}
\end{align}
where we use (\ref{Th2_bound2}) in (III-k).

To control the IP constraint $\eta$ at the MU, the AP of the C-SBS needs to satisfy $0\leq \rho_{\text{AP}}\leq 1$ and $\rho_{\text{AP}}\rho_{\text{II}}\frac{S_{\text{c}}}{S_{\text{II}}} \leq \eta$, i.e., $\rho_{\text{AP}} \leq \min \left\{\frac{\eta S_{\text{II}}}{\rho_{\text{II}}S_{\text{c}}},1\right\}=\min \left\{\frac{\eta \left[\pi R^2-\pi(d_1-r)^2\right]}{\rho_{\text{II}}S_{\text{c}}},1\right\}$, which can be lower bounder by
\begin{align} \nonumber
&\min \left\{\frac{\eta \left[\pi R^2-\pi(d_1-r)^2\right]}{\rho_{\text{II}}S_{\text{c}}},1\right\}\\
\overset{\text{(III-l)}}{\geq} & \min \left\{\frac{\eta 2^K \left[\pi R^2-\pi (d_1-r)^2\right]}{S_{\text{c}}},1\right\},
\label{III_VI_bound2}
\end{align}
where we use (\ref{III_VI_bound1}) in (III-l).

Thus, to protect the MU in this case, the maximum AP of the C-SBS is
\begin{align}
\rho_{\text{AP}}=\min \left\{\frac{\eta 2^K \left[\pi R^2-\pi(d_1-r)^2\right]}{S_{\text{c}}},1\right\},
\label{III_VI_AP}
\end{align}
where $S_{\text{c}}$ is shown in (\ref{S_c_SIII}).

\subsection{Derivation of $S_{\text{c}}$ in Scenario III }

            \begin{figure}[t!]
            \centering
            \includegraphics[scale=0.6]{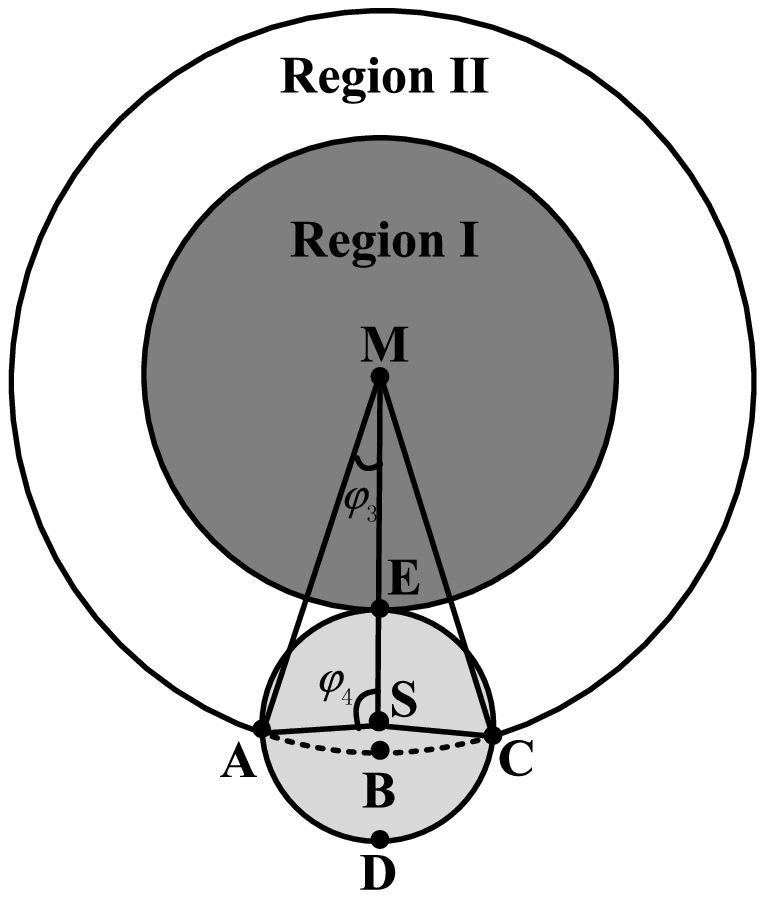}
            \caption{Geometrical model to calculate $S_{\text{c}}$ in Scenario III, i.e., $R-r\leq d_1 \leq R+r$. }
            \label{p_c_edge}
        \end{figure}

In this part, we consider that the MU is in Scenario III and calculate the area $S_{\text{c}}$. Without loss of generality, we shown the geometrical model in Fig. \ref{p_c_edge}. To begin with, we draw auxiliary segments MA, MC, SA, SC, and MS. Then, we have MA=MC=$R$, SA=SC=$r$, and MS=$d_1$. Then, we denote ``B'' as a point both on the edge of $\text{M}(R)$ and in $\text{S}(r)$, denote ``D'' as a point which is on the edge of $\text{S}(r)$ and outside of $\text{M}(R)$, and denote ``E'' as the intersection between $\text{S}(r)$ and $\text{M}(d_1-r)$. Besides, we denote $\angle \text{AMS}=\varphi_3$ and $\angle \text{ASM}=\varphi_4$. According to the cosine theorem, we have
\begin{equation}
r^2=R^2+d_1^2-2Rd_1\cos \varphi_3
\end{equation}
and
\begin{equation}
R^2=r^2+d_1^2-2rd_1\cos \varphi_4.
\end{equation}

Thus, we obtain $\varphi_3$ and $\varphi_4$ as
\begin{equation}
\varphi_3=\arccos \frac{R^2+d_1^2-r^2}{2Rd_1}
\end{equation}
and
\begin{align}
\varphi_4= \arccos \frac{r^2+d_1^2-R^2}{2rd_1}.
\end{align}
Then, we have
 \begin{equation}
S_{\text{c}}=S_{\text{ABCE}}=S_{\text{ABCM}}+S_{\text{ASCE}}-S_{\text{ASCM}},
\end{equation}
 where $S_{\text{ABCM}}=\varphi_3 R^2$, $S_{\text{ASCE}}=\varphi_4 r^2$, and $S_{\text{ASCM}}=Rd_1 \sin \varphi_3$.

Thus, we calculate $S_{\text{c}}$ in this scenario as (\ref{S_c_SIII}).

\end{document}